%% file: main.tex
\documentclass[sigconf, nonacm]{acmart}
\usepackage{hyperref}
\usepackage{enumitem}
\usepackage{tabularx} 
\usepackage[linesnumbered,vlined,ruled]{algorithm2e}
\usepackage{multirow}
\usepackage{graphicx}
\usepackage{caption}
\usepackage{subcaption}
\usepackage{booktabs} 
\usepackage{geometry} 
\usepackage{hyperxmp}
\usepackage{pgfplots}
\usepackage{placeins}
\usepackage{booktabs}
\usepackage[linesnumbered,ruled,vlined]{algorithm2e}
\usepackage{afterpage}
\usepackage{titlesec}
\titlespacing*{\subsection}{0pt}{*1}{*0.5}
\usepackage{float}
\newcommand\vldbdoi{XX.XX/XXX.XX}

\newcommand\vldbavailabilityurl{https://github.com/HCAHOI/RapidStore}

\begin{document}
\title{RapidStore: An Efficient Dynamic Graph Storage System for Concurrent Queries}

\author{Chiyu Hao$^1$, Jixian Su$^1$, Shixuan Sun$^1$, Hao Zhang$^2$, Sen Gao$^1$, Jianwen Zhao$^2$, Chenyi Zhang$^2$, Jieru Zhao$^1$, Chen Chen$^1$, Minyi Guo$^1$}
\affiliation{%
  \institution{$^1$Shanghai Jiao Tong University, China}
  \institution{$^2$Huawei, China}
  \country{}
}
\email{%
  {hcahoi11, sjx13623816973, sunshixuan, zhao-jieru, chen-chen, guo-my}@sjtu.edu.cn, {zhanghao687, zhaojianwen3, zhangchenyi2}@huawei.com, Gawssin@gmail.com
}








\input{0_abstract}

\maketitle



\newcommand{\sun}[1]{{\textcolor{black}{#1}}}
\newcommand{\new}[1]{{\textcolor{black}{#1}}}
\newcommand{\todo}[1]{{\textcolor{red}{TODO #1}}}

\input{1_introduction}
\input{2_preliminaries}
\input{3_motivation}
\input{4_design}

\input{5_concurrency}
\input{6_data_storage}

\input{7_experiment}

\input{9_related_work}
\input{10_conclusion}


\bibliographystyle{ACM-Reference-Format}
\bibliography{reference}

\input{11_appendix}

\end{document}

%% file: 0_abstract.tex
\begin{abstract}

Dynamic graph storage systems are essential for real-time applications such as social networks and recommendation, where graph data continuously evolves. However, they face significant challenges in efficiently handling concurrent read and write operations. We find that existing methods suffer from write queries interfering with read efficiency, substantial time and space overhead due to per-edge versioning, and an inability to balance performance, such as slow searches under concurrent workloads. To address these issues, we propose RapidStore, a holistic approach for efficient in-memory dynamic graph storage designed for read-intensive workloads. Our key idea is to exploit the characteristics of graph queries through a decoupled system design that separates the management of read and write queries and decouples version data from graph data. Particularly, we design an efficient dynamic graph store to cooperate with the graph concurrency control mechanism. Experimental results demonstrate that RapidStore enables fast and scalable concurrent graph queries, effectively balancing the performance of inserts, searches, and scans, and significantly improving efficiency in dynamic graph storage systems.

\end{abstract}

%% file: 1_introduction.tex
\section{Introduction}\label{sec:introduction}

As graph data frequently evolves, in-memory dynamic graph storage systems play a crucial role in applications such as social networks~\cite{gupta2014real, song2019session, sharma2016graphjet}, transaction networks~\cite{qiu2018real}, and transportation networks~\cite{rathore2015efficient}. These systems aim to efficiently support concurrent read and write queries on graphs, enabling real-time online graph data processing~\cite{sahu2017ubiquity}. Given the importance of dynamic graph storage, several systems have been proposed recently such as Sortledton~\cite{fuchs2022sortledton}, Teseo~\cite{de2021teseo}, and LiveGraph~\cite{zhu2019livegraph} to support transactional graph queries that ensure serializability of concurrent operations. Specifically, since the \emph{Compressed Sparse Row} (CSR) graph format cannot efficiently handle updates, these works focus on optimizing graph storage formats to facilitate both read and write operations.

To ensure isolation among concurrent queries, these methods adopt concurrency control mechanisms designed for relational database management systems (RDBMS) in the graph context~\cite{larson2011high, neumann2015fast}. Specifically, they use \emph{Multi-Version Concurrency Control} (MVCC) to maximize parallelism. As edges are the basic units in graphs, these systems employ a \emph{per-edge versioning} strategy, maintaining multiple versions for each edge so that different queries can simultaneously access the same edge. Furthermore, since graph queries typically follow a "vertex-neighbor" access pattern---first locating a vertex $u$ and then operating on its neighbor set $N(u)$---these systems acquire locks on vertices rather than edges to coordinate access.

Despite these advancements, we observe significant performance issues due to their concurrency control methods. First, both read and write operations must acquire a lock on a vertex before accessing it, leading to lock contention and performance degradation for read queries. Second, graph queries are generally read-intensive and heavily rely on scan operations that traverse neighbor sets, such as in PageRank and breadth-first search. Consequently, the overhead of version checks on each edge access is substantial, and maintaining versions for each edge results in significant memory consumption. Third, while existing systems focus on optimizing scan performance for traversal-based queries, they overlook search efficiency, failing to support complex queries like pattern matching. These challenges highlight the need for improved concurrency control, version management, and graph data structure design.

To address these challenges, we propose \textbf{RapidStore}. This holistic solution, which is designed for read-intensive applications, combines a novel graph concurrency control mechanism with an optimized graph data store. Specifically, we propose a \emph{subgraph-centric graph concurrency control} mechanism, which maintains versions at the subgraph level. This coarse-grained version management separates version information from the graph data, eliminating the cost of version checks. We coordinate write queries using the MV2PL protocol, while read queries operate on graph snapshots without any locks, mitigating the impact of writes on read efficiency. To enhance the concurrency control mechanism and enable fast search operations, we design the \emph{Compressed Adaptive Radix Tree} (C-ART) to store graphs. Extensive experiments demonstrate that RapidStore achieves up to 3.46x speedup over the latest systems on graph analytic queries while saving 56.34\% of memory. RapidStore also exhibits high concurrency: with 4 writers and 28 readers running concurrently, its read query completion time increases by at most 13.36\%, compared to up to 41.04\% for other systems. In summary, this paper makes the following contributions:


\begin{itemize}[noitemsep,topsep=0pt,leftmargin=*]
    \item Introduction of a novel graph concurrency control mechanism that minimizes read-write interference.
    \item Development of an optimized graph data store that balances performance across various operations.
    \item Implementation of a decoupled system design, enhancing overall system responsiveness.
\end{itemize}

Compared to existing approaches, RapidStore aims to achieve a high concurrent performance while achieving high read performance and comparable write performance, thus fulfilling the unmet needs of current dynamic graph storage systems. 

%% file: 2_preliminaries.tex
\section{Preliminary} \label{sec:background}

We focus on directed graphs, denoted as \(G = (V, E)\), where \(V\) represents the set of vertices and \(E \subseteq V \times V\) denotes the set of directed edges. An edge \(e(u, v) \in E\) indicates a directed connection from \(u\) to \(v\). For any vertex \(u \in V\), we define \(N(u)\) as the neighborhood of \(u\), which includes all vertices \(v\) such that \(e(u, v) \in E\). The out-degree of vertex \(u\), denoted as \(d(u)\), is given by \(|N(u)|\). In contrast, an undirected graph can be represented by storing edges in both directions, i.e., \(e(u, v)\) and \(e(v, u)\). Each vertex is assigned a unique integer ID within the range \([0, |V|)\). Previous studies~\cite{aberger2017emptyheaded, han2018speeding, yang2024hero, wang2017experimental, gonzalez2012powergraph} have shown that using integer vertex IDs within this range offers significant computational and storage advantages due to their compact representation. Therefore, existing works~\cite{zhu2019livegraph, macko2015llama, de2021teseo, fuchs2022sortledton, dhulipala2019low} either require vertex IDs to fall within \([0, |V|)\) or utilize dictionary encoding techniques to map external vertex IDs to this range. In this paper, we assume that each vertex ID \(u \in V\) is an integer in \([0, |V|)\). An ID is an 4-bytes integer in the implementation.

A \emph{dynamic graph} \(G = (G_0, \Delta \mathcal{G})\) represents the evolution of a graph over time, where \(G_0\) is the initial graph and \(\Delta \mathcal{G}\) is a sequence of updates. Each update \(\Delta G_i = (\oplus, e)\) modifies the graph by either inserting or deleting \(e(u, v)\), with \(\oplus = +/-\) indicating the operation. By applying updates from \(\Delta G_0\) to \(\Delta G_i\), we obtain the graph \(G_i\). Dynamic graph storage systems are designed to efficiently support both read and write operations: \textit{Search(u, v)}, which finds vertex \(v\) in \(N(u)\), \textit{Scan(u)}, which traverses \(N(u)\), and \textit{Insert(u, v)} and \textit{Delete(u, v)}, which add or remove a neighbor from \(N(u)\) (i.e., adding or removing an edge \(e(u, v)\)). \sun{Since graph applications are typically read-intensive, this paper focuses on scenarios with small updates and heavy read operations, consistent with prior works~\cite{fuchs2022sortledton,zhu2019livegraph,de2021teseo}. Nevertheless, RapidStore also efficiently supports batch updates involving tens of thousands of edges.}

\vspace{2pt}
\noindent\textbf{Adaptive Radix Tree.} The \emph{adaptive radix tree} (ART)~\cite{leis2013adaptive} is a high-performance data structure optimized for fast and memory-efficient key storage and retrieval. It builds on the radix tree by indexing keys through byte sequences. In a typical radix tree, each node may hold up to 256 child pointers (one for each byte value), which enables rapid lookups but can waste memory when many pointers are unused. ART overcomes this inefficiency by dynamically adapting the size of its nodes to the number of children. It defines four node types: \textit{N4}, \textit{N16}, \textit{N48}, and \textit{N256}, supporting up to 4, 16, 48, and 256 child pointers, respectively. Nodes automatically grow or shrink as keys are inserted or deleted, optimizing memory usage. Furthermore, ART employs \emph{path compression} to reduce tree depth by compressing common prefixes shared by multiple keys into a single edge or node. \emph{Lazy expansion} further collapses nodes by removing the path to single leaf. These two techniques reduce the number of nodes, speeding up searches. For $n$ elements in ART, each with a byte length of $w$, the time complexity for \emph{Search}, \emph{Insert}, and \emph{Delete} operations is $O(w)$, while \emph{Scan} takes $O(n)$. The space complexity is $O(n)$.

\setlength{\textfloatsep}{0pt}
\begin{figure}[t]\small
    \setlength{\abovecaptionskip}{3pt}
    \setlength{\belowcaptionskip}{0pt}
    \includegraphics[scale=0.65]{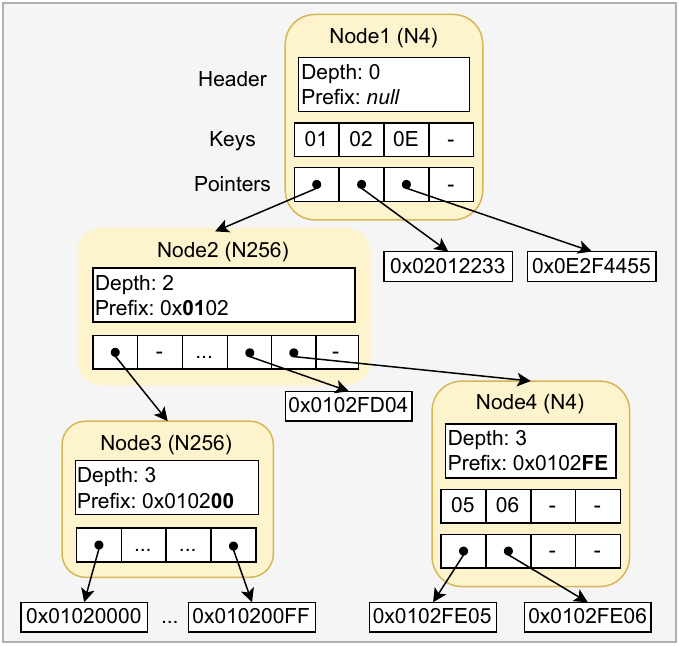}
    \centering
    \caption{An example of ART.}
    \label{fig:art}
\end{figure}

\begin{example}

Figure \ref{fig:art} shows an ART storing a set of vertices represented by 4-byte integers (8 hexadecimal digits). Each leaf node holds a vertex. \textit{N4} nodes store up to four key-pointer pairs, while \textit{N256} nodes store pointers only, with the array index implicitly indicating the key. \emph{Depth} $D$ tracks which byte in the sequence is used as the key for the node, while \emph{Prefix} records the common prefix from the root to that node. The byte sequence $M$ of a vertex is indexed from 0. To search vertex $M = 0x010200FF$, we first examine $M[0] = 01$ in \emph{Node1} since \emph{Node1}'s depth is 0, and follow the pointer to \emph{Node2}. Due to path compression, \emph{Node2}'s depth is 2, so we skip to $M[2] = 00$ and proceed to \emph{Node3}, which has a depth of 3. We match $M[3] = FF$ in \emph{Node3} and retrieve the target vertex. \emph{Insert} and \emph{Delete} operations follow similar steps to add or remove elements. \emph{Scan} performs a depth-first search across all values.

\end{example}

\vspace{2pt}
\noindent\textbf{Graph Concurrency Control.} We use Sortledton~\cite{fuchs2022sortledton} as a representative to describe existing graph concurrency control mechanisms. Sortledton employs MVCC to manage concurrent operations, maintaining multiple version logs for each edge in a linked list, called the version chain, ordered from newest to oldest. Each log entry records the operation type and its timestamp. To ensure data consistency during concurrent reads and writes, Sortledton adapt MV2PL: Writers acquire locks on all vertices they will modify at the beginning of the transaction and release them after completing the updates, while readers lock the accessed vertices and release them immediately after use. This approach leverages the characteristic that the set of vertices affected by a write query is known at the start of the transaction. 

%% file: 3_motivation.tex
\section{Motivation} \label{sec:motivation}



%

Although several dynamic graph storage methods have been proposed recently, we observe that they have severe performance issues, leading to unique challenges of efficiently processing concurrent read and write queries. In the following, we demonstrate these issues using Sortledton~\cite{fuchs2022sortledton}, the state-of-the-art method. For brevity, we present the experiment results on \emph{livejournal} (\emph{lj}) and \emph{graph500} (\emph{g5}). The server has a CPU equipped with 32 physical cores. For detailed experimental settings, please refer to Section \ref{sec:experiments}.

\vspace{2pt}
\noindent\textbf{Issue 1: Interference between Concurrent Reads and Writes.} In existing approaches, both read and write operations must acquire a lock on a vertex before accessing its neighbor set or properties to ensure serializability among queries. This locking strategy leads to contention when multiple queries attempt to access the same vertex concurrently, causing performance degradation for both reads and writes. This issue is particularly severe for high-degree vertices, accessed more frequently due to their numerous connections.

\setlength{\textfloatsep}{0pt}
\begin{figure}[t]
	\setlength{\abovecaptionskip}{0pt}
	\setlength{\belowcaptionskip}{-10pt}
		\captionsetup[subfigure]{aboveskip=0pt,belowskip=0pt}
	\centering
	\begin{subfigure}[t]{0.23\textwidth}
		\centering
		\includegraphics[width=\textwidth]{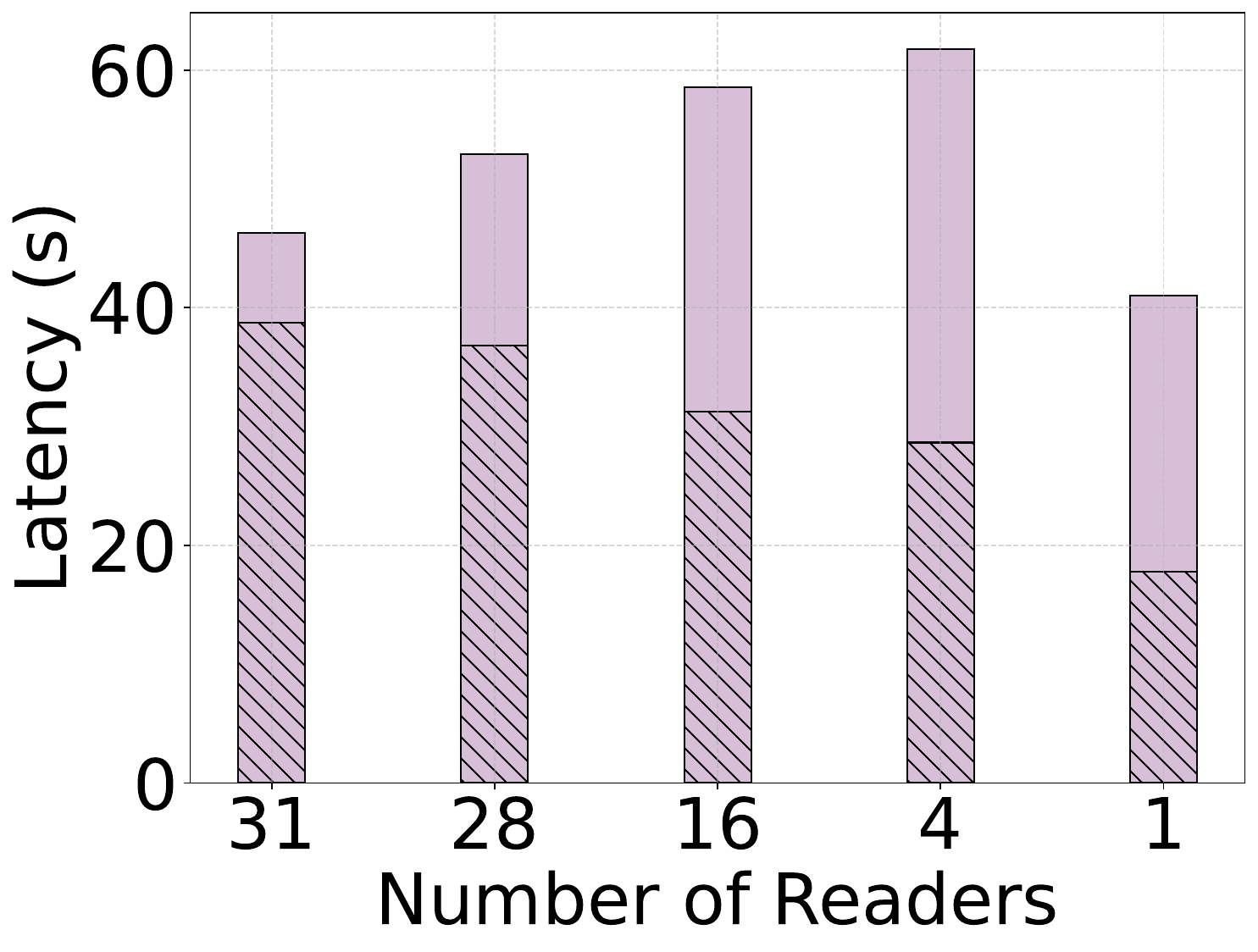}
		\caption{PageRank on \emph{lj}.}
	\end{subfigure}
        \begin{subfigure}[t]{0.23\textwidth}
		\centering
		\includegraphics[width=\textwidth]{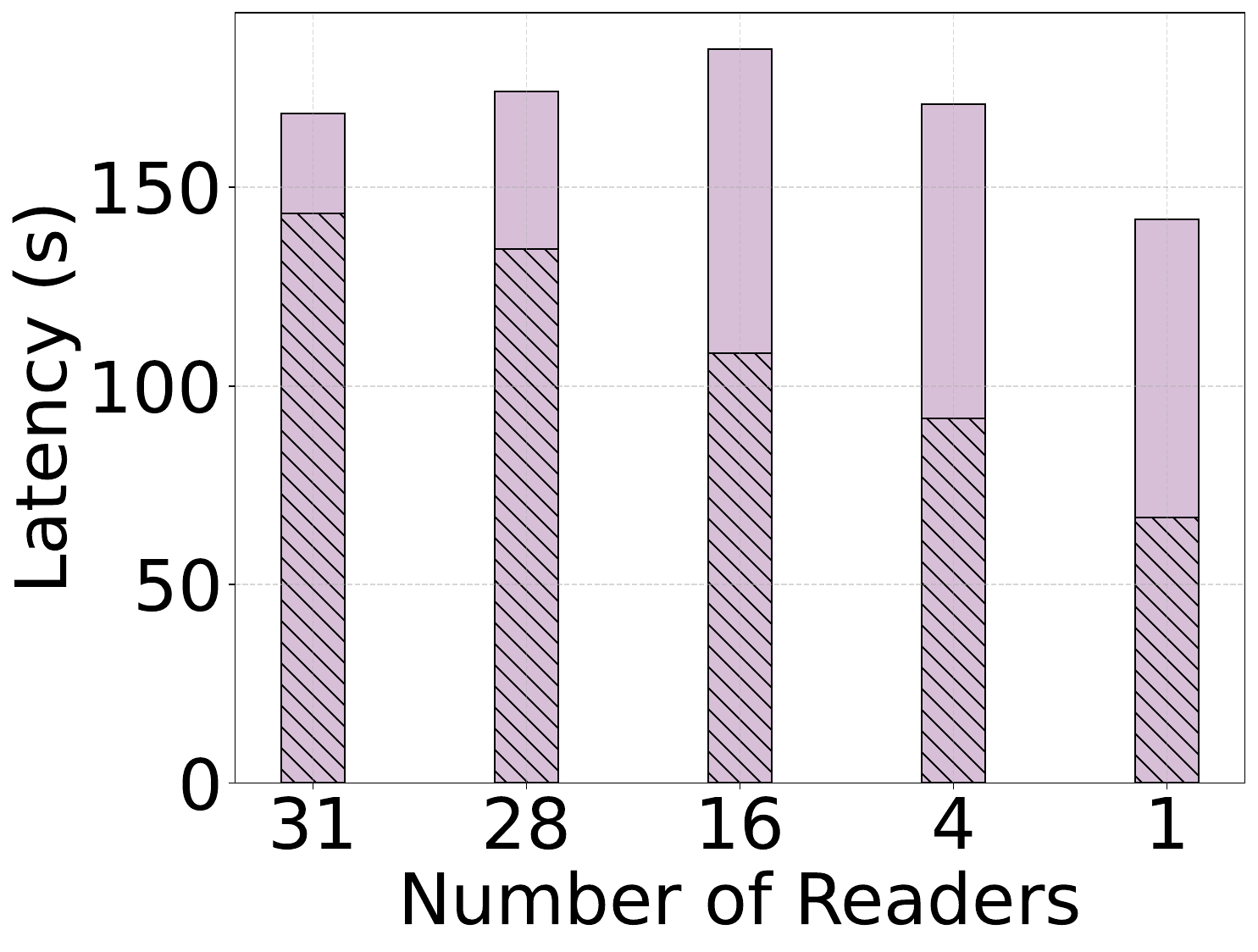}
		\caption{PageRank on \emph{g5}.}
	\end{subfigure}
	\caption{Performance under varying numbers of readers and writers (total threads fixed at 32). The shadowed bars represent the latency of readers in the absence of writers.}
    \label{fig:concurrent_read_sortledton}
\end{figure}

\setlength{\textfloatsep}{3pt}
\begin{figure}[t]
	\setlength{\abovecaptionskip}{0pt}
	\setlength{\belowcaptionskip}{0pt}
		\captionsetup[subfigure]{aboveskip=0pt,belowskip=0pt}
	\centering
	\begin{subfigure}[t]{0.23\textwidth}
		\centering
		\includegraphics[width=\textwidth]{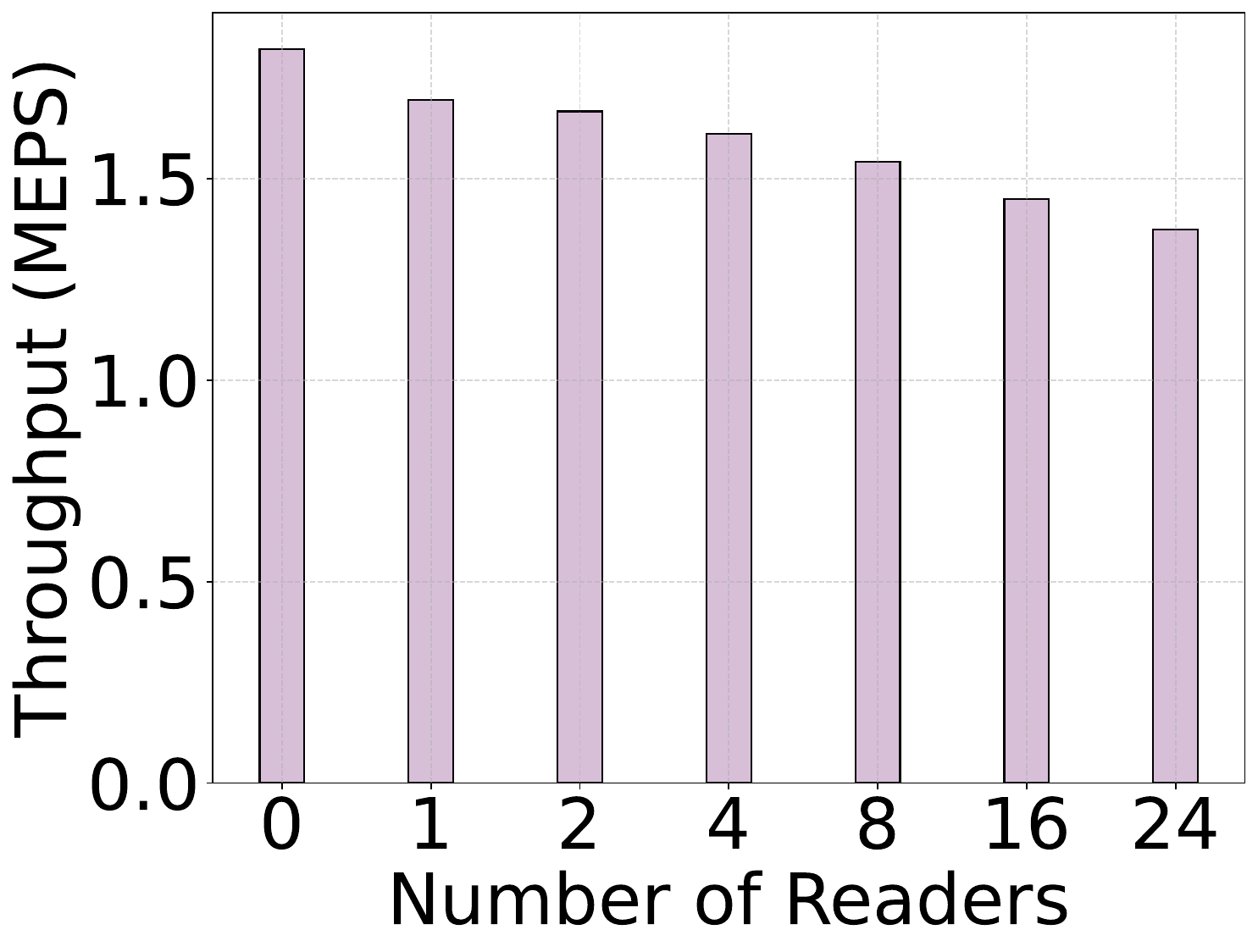}
		\caption{Insertion on \emph{lj}.}
	\end{subfigure}
        \begin{subfigure}[t]{0.23\textwidth}
		\centering
		\includegraphics[width=\textwidth]{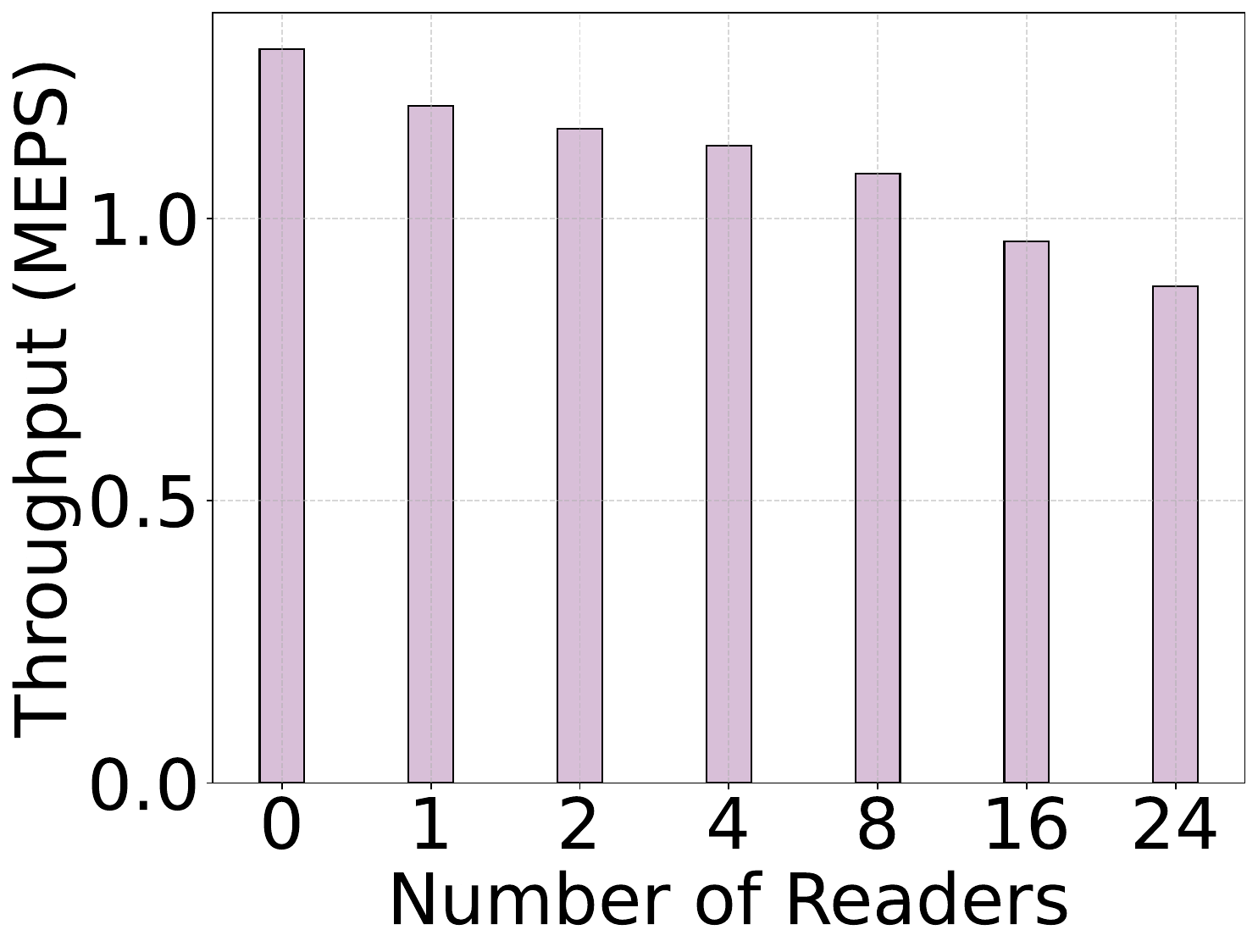}
		\caption{Insertion on \emph{g5}.}
	\end{subfigure}
	\caption{Insertion throughput as the number of readers varies, with the number of writers fixed at 8.}
    \label{fig:concurrent_write_sortledton}
\end{figure}

To illustrate this issue, Figure \ref{fig:concurrent_read_sortledton} shows the latency for 10 iterations of PageRank. Each reader independently executes a PageRank query. When no writers are active, the average latency of readers decreases as the number of readers reduces. This reduction occurs because fewer readers alleviate memory bandwidth pressure during the frequent graph scans in PageRank. Conversely, when writers are performing updates, read performance significantly deteriorates due to lock contention between read and write operations. Figure \ref{fig:concurrent_write_sortledton} examines insertion throughput as the number of readers increases while keeping the number of writers fixed at 8. The results demonstrate a drop in insertion throughput as more readers are added. This decline stems from increased lock contention, which impairs the writers’ ability to execute updates efficiently. These findings highlight a significant challenge:

\emph{\textbf{Challenge 1:}} How can we mitigate the interference between concurrent read and write queries to enhance performance and scalability in graph storage systems?

\vspace{2pt}
\noindent\textbf{Issue 2: Time and Space Overheads Due to Per-Edge Versioning.} To support concurrent read and write queries, existing methods maintain versions for each edge. This strategy requires queries to perform a version check on each edge access to ensure they retrieve the correct data version. Additionally, storing version chains for each edge increases memory requirements.

\begin{table}[h]
    \centering
     \small
     \setlength{\abovecaptionskip}{0pt}
     \setlength{\belowcaptionskip}{0pt}
    \caption{Performance of \emph{Search}, \emph{Scan} and PageRank (PR) without and with versioning.}
    \label{tab:version_cost}
   
\begin{tabular}{|c|c|c|c|c|}
\hline
 \textbf{Dataset} & \textbf{Versioned?} & \textbf{Search (TEPS)} & \textbf{Scan (TEPS)} & \textbf{PR (s)} \\ \hline
\multirow{2}{*}{\emph{lj}} & No & \textbf{2270.85} & \textbf{40371.23} & \textbf{16.23} \\
 & Yes & 1898.17 & 19033.40 & 31.16 \\ \hline
\multirow{2}{*}{\emph{g5}} & No & \textbf{672.89} & \textbf{89055.10} & \textbf{44.15} \\
 & Yes & 635.12 & 38885.92 & 94.26 \\ \hline
\end{tabular}
\end{table}

Table \ref{tab:version_cost} compares the performance of Sortedlton with and without versioning. Scan iterates over the neighbor set of a vertex, while Search identifies a target edge. Both are measured in thousand edges processed per second (TEPS). Versioning reduces \emph{Scan} throughput by approximately 53\% for \emph{lj} and 56\% for \emph{g5}, reflecting the overhead of version checks during frequent edge scans. For \emph{Search}, the throughput drops by about 16\% for \emph{lj} and 6\% for \emph{g5}, showing a smaller but noticeable impact. PageRank experiences a significant performance degradation, with latency nearly doubling for both datasets. This issue highlights the second challenge.

\emph{\textbf{Challenge 2:}} How can we reduce the time and space inefficiencies caused by per-edge versioning to enhance the speed and scalability of graph analytics?

\vspace{2pt}
\noindent\textbf{Issue 3: Inefficient Support for Fast Search Operations.} Traversal-based algorithms like PageRank and breadth-first search rely heavily on fast scan operations. Consequently, existing methods prioritize optimizing scan performance while often neglecting the efficiency of search operations, which are critical for pattern matching tasks such as triangle counting (TC). For instance, when intersecting two sets of vastly different sizes, iterating over the smaller set and performing searches in the larger set is more efficient than using merge-based set intersections.

\begin{table}[h]
\vspace{-10pt}
\centering
  \small
     \setlength{\abovecaptionskip}{0pt}
     \setlength{\belowcaptionskip}{0pt}
\caption{Performance of \emph{Search} and Triangle Counting (TC).}
\label{tab:search_and_tc}
\resizebox{\columnwidth}{!}{
\begin{tabular}{|c|c|ccc|c|}
\hline
\multirow{2}{*}{\textbf{Dataset}} & \multirow{2}{*}{\textbf{Method}} & \multicolumn{3}{c|}{\textbf{Search (TEPS)}} & \multirow{2}{*}{\textbf{TC (s)}} \\ \cline{3-5}
 &  & \textbf{General} & \textbf{Low Deg.} & \textbf{High Deg.} &  \\ \hline
\multirow{2}{*}{\emph{lj}} & CSR & \textbf{7116.73} & \textbf{8567.41} & \textbf{6745.4} & \textbf{47.21} \\
 & Sortledton & 1898.65 & 2139.71 & 2414.32 & 153.30 \\ \hline
\multirow{2}{*}{\emph{g5}} & CSR & \textbf{3342.92} & \textbf{7437.11} & \textbf{3031.4} & \textbf{5268.23} \\
 & Sortledton & 635.33 & 1830.93 & 546.09 & 18786.13 \\ \hline
\end{tabular}
}
\vspace{-10pt}
\end{table}

Table \ref{tab:search_and_tc} presents results comparing the performance of search operations and TC between the CSR format and Sortedlton. The results indicate that Sortedlton is approximately 3.26x slower than CSR for search operations and 3.12x slower for TC. This significant performance gap highlights the following challenge:

\textbf{Challenge 3:} How can we design an efficient graph data store that enables fast search operations while maintaining a balance across scan and insert operations?

%% file: 4_design.tex
\setlength{\textfloatsep}{0pt}
\begin{figure*}[t]
    \setlength{\abovecaptionskip}{3pt}
    \setlength{\belowcaptionskip}{0pt}
    \includegraphics[scale=0.6]{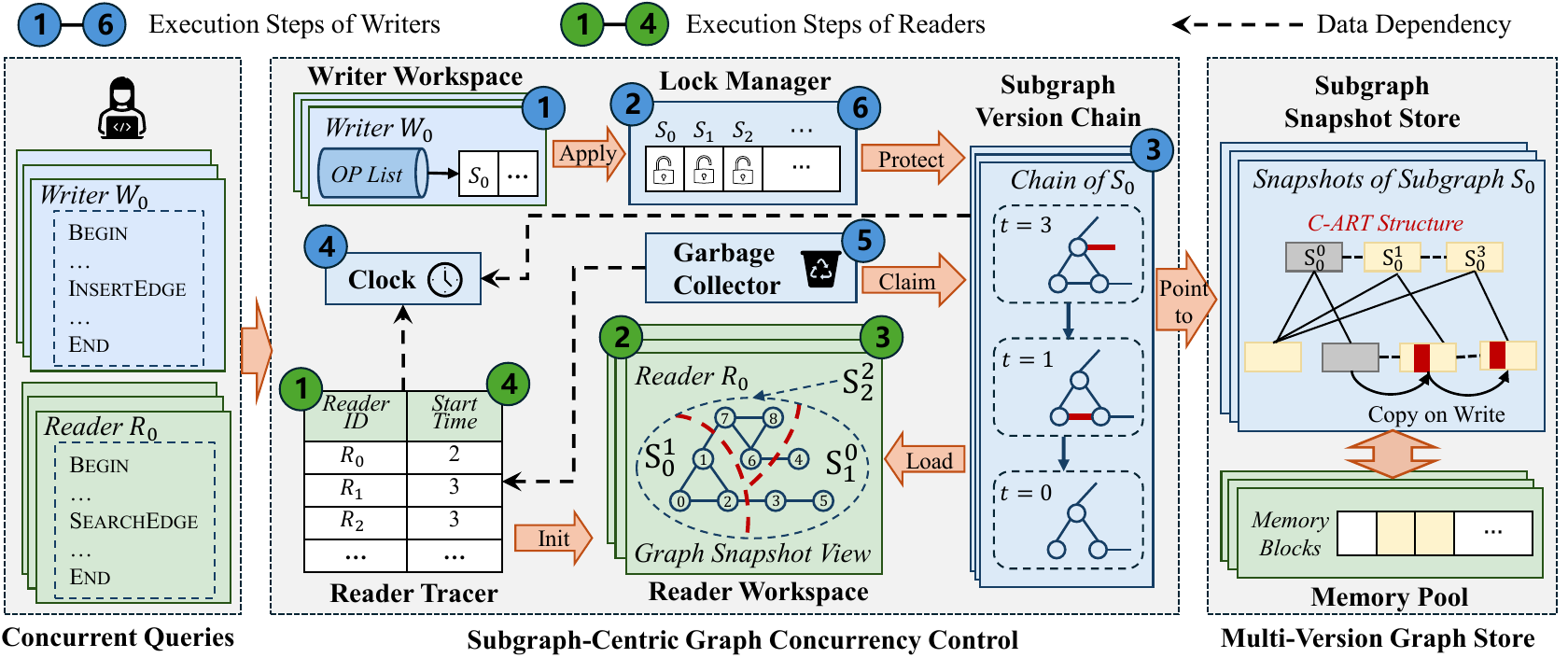}
    \centering
    \caption{An overview of RapidStore.}
    \label{fig:system_overview}
\vspace{-10pt}
\end{figure*}

\section{An Overview of RapidStore}\label{sec:system_design}



%

The significant issues and challenges faced by existing methods severely limit the ability to perform fast and scalable concurrent graph queries. To address these problems, we propose \textbf{RapidStore}, a holistic solution that combines a novel graph concurrency control mechanism with an optimized graph data store.

\vspace{2pt}
\noindent\textbf{Execution Flow.} Figure \ref{fig:system_overview} provides an overview of RapidStore. \sun{The system uses a logical clock, initialized to 0, to coordinate query execution and versioning.} Since write queries update the graph with a priori known write sets, RapidStore coordinates them using the classical MV2PL protocol to ensure their serializability.

Specifically, RapidStore executes a write query in six steps:
\textcircled{1} Identify the set of subgraphs $S$ to be modified based on the write set. \textcircled{2} Obtain locks on these subgraphs to ensure exclusive access during updates among write queries. \textcircled{3}  Create new versions of the subgraphs using the copy-on-write strategy, assigning each a version number equal to the current clock value plus 1. \textcircled{4} Commit the changes and increment the clock by 1 to represent the new state of the graph. \textcircled{5} Remove outdated versions based on the read queries in execution to free up resources. \textcircled{6} Release the locks held by the write query, allowing other queries to access the subgraphs.

In contrast, RapidStore executes a read query as follows: \textcircled{1}  Register the query with the start time obtained from the current clock value. \textcircled{2} Build the graph snapshot by selecting the appropriate versions of subgraphs based on the start time. \textcircled{3} Perform graph data access operations on this snapshot to complete the query. \textcircled{4} Unregister the query after execution is complete. \sun{The initial graph $G_0$ is associated with version 0, indicated by the initial value of the logical clock. As a result, the reader execution flow remains consistent regardless of whether updates have occurred, i.e., whether the clock is 0 (no updates) or greater than 0 (some updates committed).}

\vspace{2pt}
\noindent\textbf{System Design.} To support this concurrency control method, our key idea is to exploit the characteristics of graph queries through a decoupled system design that separates the management of read and write queries, maintains coarse-grained versions at the subgraph level, and decouples version data from graph data. Specifically, a write query creates a new snapshot for the subgraphs it modifies using the copy-on-write strategy instead of creating a version at the edge level, while read queries construct a graph snapshot by combining snapshots of these subgraphs. Additionally, RapidStore employs a memory pool to support the copy-on-write strategy, reducing the overhead of frequent memory allocation and deallocation by the operating system.

\vspace{2pt}
\noindent\textbf{Advantages and Novelty.} The novel decoupled design offers significant advantages over traditional edge-versioning approaches:

\emph{Non-Blocking Reads through Decoupled Query Management.} While existing approaches like Sortledton and Teseo also manage to write and read queries differently, they require writers to acquire exclusive locks and readers to acquire shared locks on the neighbor sets of vertices to be accessed, leading to blocking between them. In contrast, RapidStore’s separation of read and write queries ensures that read queries do not require any lock operations after registration. This novel approach allows read queries to execute efficiently without interference from concurrent writes, and vice versa.

\emph{Efficient Snapshot Retrieval with Coarse-Grained Subgraph Versioning.} Read queries need to work on a graph snapshot. Traditional edge-versioning strategies are inadequate for fast snapshot retrieval because they require scanning the entire graph to find and store the correct version for each edge, incurring prohibitive time and space costs. By maintaining versions at the subgraph level rather than per-edge, RapidStore enables fast graph snapshot retrieval by combining snapshots of subgraphs without scanning and storing each edge individually. This strategy significantly reduces the time and space costs associated with version maintenance.

\emph{Elimination of Version Checks via Decoupled Version and Graph Data.} In traditional approaches, version and graph data are stored together, requiring both to be loaded from memory to the CPU and frequent version checks during graph operations, which adds significant overhead. By decoupling version data from graph data and utilizing coarse-grained versioning, RapidStore eliminates the need for both read and write queries to check versions during running time, reducing operational overhead and improving performance.

Additionally, RapidStore’s optimized graph data store supports the above graph data management methodology and effectively balances scan, search, and insert operations, providing consistent efficiency across different types of graph workloads. These innovations collectively allow RapidStore to overcome the limitations of traditional edge-versioning approaches, offering a more efficient and scalable solution for dynamic graph storage systems. In the following sections, we will detail the subgraph-centric concurrency control and the multi-version graph store.

%% file: 5_concurrency.tex
\section{Subgraph-Centric Concurrency Control} \label{sec:concurrency_control}

In this section, we propose the subgraph-centric graph concurrency control mechanism.

\subsection{Coarse-Grained Version Management} \label{sec:version_management}

Given a read query $R$ starting at time $t$, $R$ needs to operate on the graph snapshot $G_t$ to ensure correctness. A straightforward method to eliminate version checks is to materialize $G_t$ at the beginning of $R$ and then execute $R$ on this materialized snapshot. However, with edge-versioning strategies, materializing $G_t$ requires scanning every edge to retrieve and store the correct versions, incurring prohibitive time and space costs.

To address this problem, we propose a coarse-grained version management approach that maintains versions at the subgraph level. Specifically, RapidStore partitions the graph $G$ by dividing its vertex set $V(G)$ into a set of equal-sized partitions $P$. Each subgraph $S$ consists of the vertex set $P$ and the edges adjacent to vertices in $P$, formally defined as $V(S) = P$ and $E(S) = \{ e(u, v) \in E(G) \mid u \in P \}$. RapidStore currently uses a simple graph partitioning strategy, which assigns continuous blocks of ${|P|}$ vertices to each partition based on their IDs. \sun{Using the graph in Figure~\ref{fig:system_overview} as an example, we set the partition size $|P|$ to 3 and partition the nine-vertex graph into three subgraphs: vertices 0–2 and their adjacent edges form $S_0$, 3–5 form $S_1$, and 6–8 form $S_2$.} 

RapidStore maintains a version chain for each subgraph, where each new version points to its predecessor as illustrated in Figure \ref{fig:system_overview}. When $S$ is modified, RapidStore creates a new version of $S$, leveraging the underlying multi-version graph store (see Section \ref{sec:data_storage}). The version chains are stored separately from the graph data, with each version keeping a pointer to its corresponding subgraph snapshot, effectively decoupling version information from graph data.


\subsection{Concurrency Control Protocol}

In existing graph storage systems, both read and write queries must acquire locks on the vertices they access, leading to interference between readers and writers and reducing overall concurrency. RapidStore decouples the management of read and write queries to alleviate this issue.

\subsubsection{Concurrency Control for Write Queries}

To synchronize write queries, we adopt MV2PL. Specifically, given a write query $W_0$ that intends to update a set of vertices $\Delta V$, we first identify the subgraphs $\Delta \mathcal{S}$ containing these vertices and thus require modification. Inspired by the locking strategy in Sortledton, we acquire locks on the subgraphs in $\Delta \mathcal{S}$ in ascending order of their subgraph IDs. This sorted locking order prevents deadlocks by ensuring a consistent lock acquisition sequence across concurrent write queries. After obtaining all the necessary locks, we update each subgraph in $\Delta \mathcal{S}$ by creating new snapshots based on the latest snapshots that reflect the changes introduced by $W_0$ (see Section ~\ref{sec:data_storage} for details).

To coordinate read and write queries and maintain consistency, RapidStore introduces two timestamps: the global write timestamp ($t_w$), which tracks the order of write query commits, and the global read timestamp ($t_r$), which indicates the latest consistent snapshot available to read queries. The commit phase for $W_0$ involves incrementing $t_w$ by 1 atomically and assigning the new value to a local commit timestamp $t$, representing the commit time of $W_0$. We then assign the commit timestamp $t$ to the new versions of the modified subgraphs, linking them to the head of their respective version chains to make them the most recent versions. Next, we poll the current value of $t_r$; if $t_r = t - 1$, we atomically increment $t_r$ by 1 to advance the global read timestamp. MV2PL ensures the serializability of write queries. The polling and conditional increment of $t_r$ guarantee that write queries commit in a serial order determined by their commit timestamps. It also ensures consistent snapshots for read queries, as read queries operating with timestamp $t_r$ access a consistent snapshot. After the commit phase, $W_0$ performs garbage collection on the modified subgraphs to remove obsolete versions (as detailed in Section~\ref{sec:garbage_collection}) and then releases all acquired locks, allowing other queries to proceed.

\sun{Take $W_0$ in Figure~\ref{fig:system_overview} as an example. Suppose it inserts edge $e(1, 6)$ into the graph when $t_w = 2$. Based on the vertex IDs, the affected subgraphs are $S_0$ and $S_2$. To avoid deadlock, $W_0$ acquires locks on $S_0$ and $S_2$ in subgraph ID order. It then creates new snapshots of both subgraphs using the copy-on-write strategy and links them to their respective version chains. Given $t_w = 2$, the new version is set to 3 by atomically incrementing $t_w$. After committing, $t_r$ is updated to 3, making the change visible to readers. Finally, $W_0$ performs GC to remove obsolete versions and releases the locks.}

\subsubsection{Concurrency Control for Read Queries}

In RapidStore, read queries do not require any locks on subgraph snapshots. Consequently, they neither block write queries nor are they blocked by them, allowing for high concurrency and performance. To keep track of active read queries, RapidStore employs a mechanism called the \emph{reader tracer}. The reader tracer is an array where each element is an 8-byte integer. The highest bit of each integer (referred to as the \emph{status} bit) indicates whether the slot is in use (1) or free (0). The remaining bits store the start timestamp of a read query. By default, the size $k$ of the array is set to the number of CPU cores in the machine, but the user can configure it to suit different workloads.

When a read query $R$ begins execution, it registers itself in the reader tracer through the following steps: 1) Loop over the reader tracer to locate an empty slot based on the status; and 2) Set the status to 1 and its start timestamp $t$ as the current read timestamp $t_r$. The operation can be executed with atomic compare-and-swap (CAS) instructions without locks, ensuring that the registration process is both efficient and thread-safe.

\sun{After registration, $R$ constructs its graph snapshot view by iterating over the version chains of all subgraphs and selecting the appropriate subgraph snapshots (the latest version of each subgraph with timestamp $t \leqslant t_r$). Pointers to these snapshots are copied into $R$’s reader workspace.} When accessing the version chain of a subgraph, $R$ does not require any locks because:

\begin{enumerate}[leftmargin=*]
\item The start timestamp $t$ ensures that only committed subgraph snapshots with versions less than or equal to $t$ are visible to $R$.
\item Writers create new subgraph snapshots using the copy-on-write strategy, which does not affect existing snapshots.
\end{enumerate}

\sun{During its execution, $R$ accesses graph data according to its constructed snapshot view, ensuring consistency without impeding concurrent write operations. Figure~\ref{fig:system_overview} illustrates the snapshot construction process. Reader $R_0$ obtains a start timestamp $t_r = 2$. For subgraph $S_0$, it traverses the version chain and selects the latest version with timestamp $t \leqslant 2$, which is $t = 1$, and copies the pointer to snapshot $S_0^1$ into its workspace. The same procedure is applied to subgraphs $S_1$ and $S_2$, completing the snapshot view.} After $R$ completes its execution, it unregisters itself by resetting the status bit in its slot in the reader tracer back to 0 and setting the start timestamp to a maximum value (e.g., the largest representable integer). This action marks the slot as available for future read queries.

\subsection{Garbage Collection} \label{sec:garbage_collection}


Each update to a subgraph creates a new snapshot, leading to an ever-growing version chain for that subgraph. Similar to garbage collection (GC) approaches in existing works~\cite{tu2013speedy, diaconu2013hekaton,zhu2019livegraph,lim2017cicada}, RapidStore reclaims stale versions based on the timestamps. Specifically, a subgraph version can be reclaimed if it is not the latest in the version chain and is not being used by any active readers. After committing its updates, a writer $W$ performs GC as follows:


\begin{enumerate}[leftmargin=*]
\item $W$ scans the reader tracer to collect the start timestamps of all active readers. To ensure it does not block other queries, $W$ retrieves each value in the reader tracer using atomic operations.
\item Based on the collected start timestamps, $W$ loops over the version chain of the modified subgraphs to determine which versions can be reclaimed.
\end{enumerate}

RapidStore leverages writers to perform GC because frequently modified subgraphs are more likely to require GC, and performing GC during write operations benefits from spatial locality without the need for additional background service threads. The garbage collection of subgraph snapshots is supported by the underlying multi-version graph store, which will be introduced in Section~\ref{sec:data_storage}.

\subsection{Analysis of Graph Concurrency Control}

Proposition \ref{prop:correctness} guarantees the correctness of the proposed concurrency control mechanism. We provide a proof sketch here and include the full details in the technical report~\cite{anonymous2024}.

\begin{proposition} \label{prop:correctness}
The subgraph-centric concurrency control mechanism guarantees the serializability of both write and read queries.
\end{proposition}

\begin{proof}
MV2PL and the enforcement of commit order based on $t_w$ and $t_r$ ensure that write queries are serialized according to their commit timestamps. Meanwhile, read queries access a consistent snapshot of the graph at their start time, proceeding without interfering with ongoing writes. Since the start timestamp depends on $t_r$, which advances based on the commit timestamps of write queries, both write and read queries are serialized in commit timestamp order. Thus, the proposition is proven.
\end{proof}

\noindent\textbf{Time and Space Cost.} We analyze the time and space overhead introduced by the graph concurrency control mechanism. Let $p = \lceil \frac{|V|}{|P|} \rceil$ denote the number of subgraphs in $G$, and let $k$ be the size of the reader tracer array corresponding to the maximum number of concurrent read queries. We have the following proposition, the proof of which is included in the technical report.

\begin{proposition} \label{prop:length_version}
For any subgraph $S$, the length of its version chain is at most $k + 1$.
\end{proposition}

For a read query $R$, the reader workspace requires $O(p)$ space to store pointers to subgraph snapshots. The concurrency control for $R$ incurs $O(p \times k)$ time overhead: 1) Registering $R$ takes $O(k)$ time, as it searches the reader tracer; 2) Constructing the snapshot view takes $O(p \times k)$ time since $R$ traverses the version chain for each subgraph; and 3) Unregistering $R$ takes $O(1)$ time. This process is efficient because new versions are at the head of the version chain, and $R$ obtains the latest value of $t_r$.

A write query $W$ can be blocked by other write queries during lock acquisition. The overhead of lock contention depends on the concurrency of write queries; here, we analyze the overhead without lock contention. Let $\Delta \mathcal{S}$ be the set of subgraphs modified by $W$, and let $s = |\Delta \mathcal{S}|$. The workspace cost of $W$ is $O(s)$. Since RapidStore sorts $\Delta \mathcal{S}$ to obtain locks, the time cost of acquiring locks is $O(s \log s)$. Committing updates takes $O(s)$ time, as $W$ sets the version value for each subgraph in $\Delta \mathcal{S}$. Finding versions to be reclaimed takes $O(k \log k + s k)$ time: $W$ acquires start timestamps of active readers, sorts them, and traverses the version chain to determine which versions to reclaim.  Since $W$ typically updates only a few subgraphs, $s$ remains small, and the version chain length $k$ is bounded as per Proposition \ref{prop:length_version}. This ensures that both commit and garbage collection operations are highly efficient.

\vspace{2pt}
\noindent\textbf{Impact of Partition Size.} The partition size (i.e., subgraph size) $|P|$ directly influences the number $k$ of subgraphs. At the minimum, $|P|$ can be set to 1, where each subgraph contains only a single vertex. This setup can improve write performance by reducing lock contention among concurrent write queries. However, as discussed earlier, smaller subgraph sizes increase the space and time overhead for read queries. Additionally, minimizing subgraph size limits opportunities to optimize the storage of small neighbor sets within a subgraph, thereby decreasing read efficiency. To balance read and write performance, we set $|P|$ to 64 by default, which achieves good empirical performance. \sun{We currently use a static partitioning strategy, randomly dividing the graph into equal-sized partitions at initialization, and dynamically selecting data structures for neighbor sets based on vertex degrees. A promising research direction is to design adaptive partitioning strategies that adjust to operation skewness, such as recent write frequencies, to further reduce write conflicts and improve write performance.}

%% file: 6_data_storage.tex
\section{Multi-Version Graph Store}\label{sec:data_storage}

We present the multi-version graph store that efficiently maintains graph data in this section.

\subsection{Design of Graph Store} \label{sec:general_idea}

\begin{figure}[t]
    \setlength{\abovecaptionskip}{0pt}
    \setlength{\belowcaptionskip}{0pt}
    \includegraphics[scale=0.68]{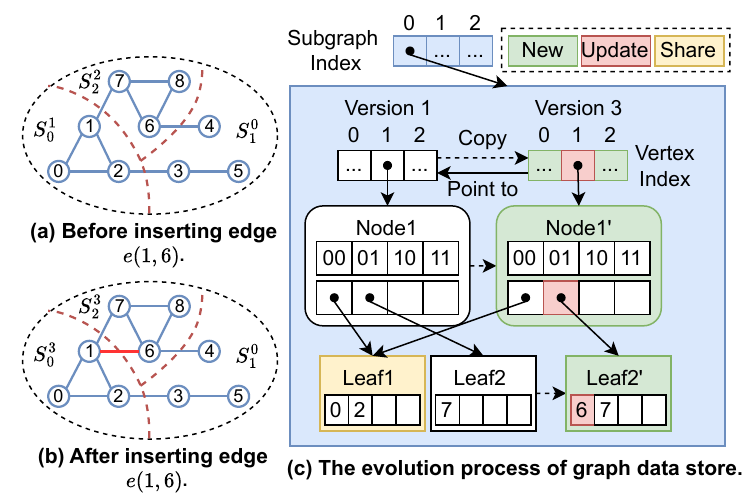}
    \centering
	\caption{\sun{Overview of the multi-version graph store design with the copy-on-write strategy.}}
	\label{fig:vertex}
\end{figure}

To efficiently handle concurrent read and write queries, a graph store must meet two key requirements. First, it must quickly generate a snapshot of a subgraph $S$ upon updates to support subgraph-centric concurrency control and ensure the correctness of concurrent queries. Second, since graph queries commonly follow a "vertex-neighbor" access pattern, the store must efficiently retrieve the neighbor set $N(u)$ for any vertex $u \in V(S)$ and support fast search, insertion, and scan operations on $N(u)$. To achieve these goals, we adopt a copy-on-write strategy for updating $S$, ensuring that read queries remain unaffected by write queries. 

\sun{Figure~\ref{fig:vertex} illustrates the design of the store. The example assumes 4-bit vertex IDs, using 2 bits per level in a classical radix tree. When $W_0$ inserts edge $e(1, 6)$ at logical time $t_w = 2$, it locates the affected subgraph $S_0$ via the subgraph index, which maintains pointers to version chains. Each subgraph version contains a vertex index pointing to radix trees that store neighbor sets. Vertices with common prefixes share the same leaf node (e.g., vertex 1’s neighbors 0 and 2 in Leaf 1). The insertion triggers a copy-on-write: the affected radix tree path (from leaf to root) is duplicated and updated, creating a new version $S_0^3$ (with $t_w$ incremented to 3). Subgraph $S_2$ is updated similarly, as $e(1, 6)$ also affects it. Its version is also set to 3, since both updates are committed in the same transaction.}

The vertex index is an array where each entry corresponds to a vertex $u \in V(S)$ and stores a pointer to its neighbor set $N(u)$. As described in Section~\ref{sec:concurrency_control}, $V(S)$ is a contiguous subset of $V(G)$ in the range $[0, |V(G)|)$, enabling $O(1)$ lookup of $N(u)$ by vertex ID. Since vertex indices are small, copying them for new versions is fast and lightweight. To handle degree skewness in real-world graphs, $N(u)$ is stored differently based on vertex degree: high-degree vertices use a tree structure, while low-degree vertices use small arrays. These small arrays are further grouped into a tree to optimize memory usage. \sun{Next, we first introduce the \emph{compressed adaptive radix tree} (C-ART) in Section 6.2, which stores the neighbor set for high-degree vertices. Then, in Section 6.3, we present the \emph{clustered index}, which centrally stores low-degree vertices within a subgraph to fully exploit locality.}



\subsection{Compressed Adaptive Radix Tree}\label{sec:c_art}

As discussed in Section \ref{sec:background}, ART is memory-efficient and supports fast retrieval and insertion operations. Its hierarchical structure, which limits each node to a maximum of 256 entries, enables an efficient copy-on-write mechanism by duplicating the root-to-leaf path. This makes ART a suitable choice for storing $N(u)$. However, storing vertices individually in each leaf can degrade scan performance due to frequent node traversal. A straightforward approach to improve scan performance is to organize the leaves into contiguous segments of size $B = 256$, each storing vertices from $N(u)$. However, the distribution of $N(u)$ is often skewed and sparse across the range $[0, |V(G)|)$, leading to low filling ratios (the proportion of occupied entries within each segment). As shown in Table \ref{tab:filling_ratio}, filling ratios for various graphs are below 4\%, resulting in poor scan performance and significant memory waste. To address this issue, we propose the \emph{compressed adaptive radix tree} (C-ART), which compresses leaves to significantly improve the filling ratio. This compression enhances memory locality and traversal efficiency, making C-ART a more effective solution for graph storage.

\begin{table}[t]
\setlength{\abovecaptionskip}{0pt}
     \setlength{\belowcaptionskip}{0pt}
\caption{Comparison of filling ratios of ART and C-ART.}
\resizebox{\columnwidth}{!}{
\begin{tabular}{|c|c|c|c|c|c|c|}
\hline
\multicolumn{1}{|l|}{\textbf{Dataset}} & \textit{lj} & \textit{ot} & \textit{ldbc} & \textit{g5} & \textit{tw} & \textit{fr} \\ \hline
ART & 2.17\% & 3.13\% & 3.44\% & 2.40\% & 2.24\% & 2.94\%\\
C-ART & \textbf{66.17\%} & \textbf{67.12\%} & \textbf{64.27\%} & \textbf{67.80\%} & \textbf{67.22\%} & \textbf{64.70\%} \\ \hline
\end{tabular}
}
\label{tab:filling_ratio}
\end{table}

\begin{figure}[t]\small
    \setlength{\abovecaptionskip}{3pt}
    \setlength{\belowcaptionskip}{0pt}
    \includegraphics[scale=0.69]{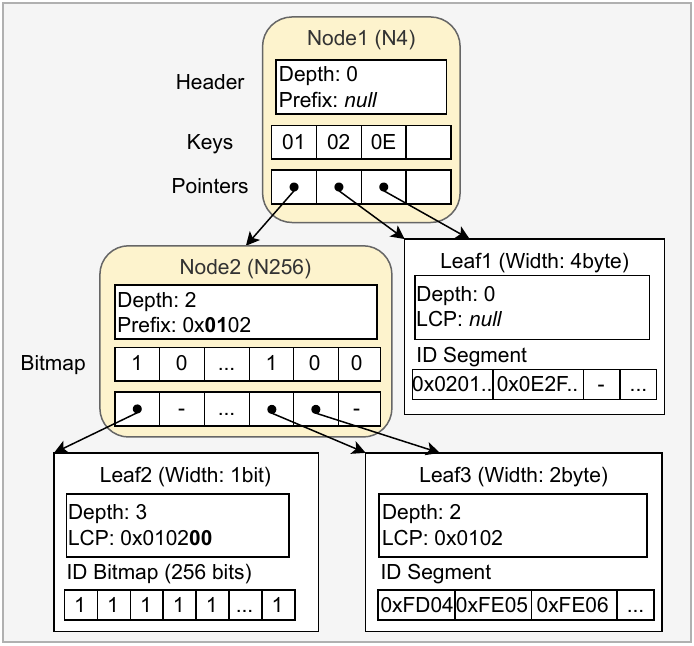}
    \centering
    \caption{An example of C-ART storing the same elements as ART shown in Figure~\ref{fig:art}.}
    \label{fig:c-art}
\end{figure}

\noindent\textbf{Compressing Leaves.} ART employs path compression by vertically merging paths and nodes within the tree, reducing memory usage and search costs. In contrast, C-ART, introduced in this paper, applies horizontal compression to leaves, enhancing scan performance in graph-related operations. Specifically, unlike ART, where each leaf stores a single vertex, C-ART leaves can store up to $B = 256$ vertices. Each C-ART leaf is defined by a \emph{longest common prefix} (LCP) shared by all vertices within the leaf, with the \emph{depth} indicating the LCP's length in bytes. This design allows multiple keys in a C-ART node to point to the same leaf, effectively compressing the leaves of multiple keys into a single leaf. This improves the filling ratio (as shown in Table \ref{tab:filling_ratio}), enhancing scan performance and reducing memory consumption.

Figure~\ref{fig:c-art} provides an example of C-ART. Each leaf contains up to 256 vertices, with the LCP and its length recorded. Multiple keys within a node can point to the same leaf (e.g., two keys in Node2 point to Leaf3). Compared to ART in Figure~\ref{fig:art}, C-ART stores vertices much more compactly. Next, we describe the graph operations on a neighbor set $N(u)$ of vertex $u$ stored in C-ART. The insert operation illustrates the construction of a C-ART.

1) \textit{Search(u, v)} locates vertex $v$ in $N(u)$ by traversing the C-ART nodes based on the byte sequence of $v$, similar to the process in ART. Upon reaching a leaf, C-ART performs a binary search to locate $v$ within the leaf. This design ensures fast search performance due to the limited number of vertices per leaf.

2) \textit{Scan(u)} traverses the C-ART structure in a depth-first-search (DFS) order, enabling sequential memory access when processing leaves. As multiple vertices are stored continuously in leaves, this approach significantly improves performance for graph analytic queries compared with ART.

3) \textit{Insert(u, v)} adds vertex $v$ to $N(u)$ stored in a C-ART. Since C-ART inherits the node management strategy of ART, the focus is on inserting $v$ into a leaf. Initially, C-ART is created with a root node and a leaf segment containing $B = 256$ entries. To insert $v$, the position in the leaf is determined using \textit{Search(u, v)}. Let the target leaf currently contain $b$ vertices. As illustrated in Figure~\ref{fig:c-art-insert}, the insertion process involves three possible cases based on the state of the leaf:

\begin{itemize}[leftmargin=*]
    \item \textbf{Case 1: $b < B$.} Insert $v$ directly into the leaf and update the affected pointers in the parent node to reflect the change.
    \item \textbf{Case 2: $b = B$ and the leaf is shared by multiple keys.} Identify the first key with a pointer offset of at least $\frac{B}{2}$. Split the leaf at this offset, then insert $v$ following the procedure in Case 1.
    \item \textbf{Case 3: $b = B$ and the leaf is associated with a single key.} Compute the LCP of all vertices in the leaf, create a new internal node using this LCP as its prefix, split the leaf as in Case 2, and insert $v$ accordingly.
\end{itemize}


\begin{figure}[t]
    \setlength{\abovecaptionskip}{3pt}
    \setlength{\belowcaptionskip}{0pt}
    \includegraphics[scale=0.68]{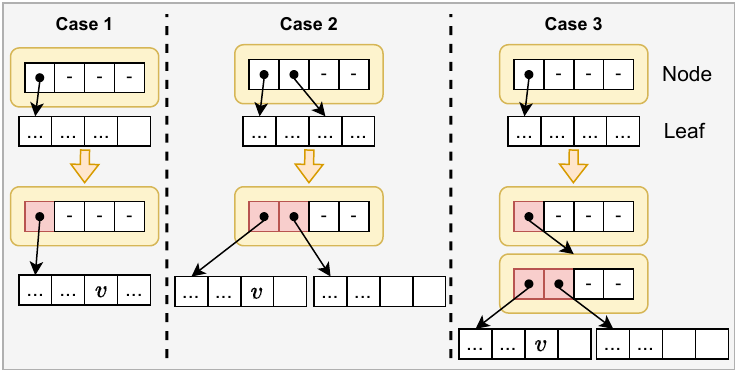}
    \centering
	\caption{Insertion of vertex $v$ into $N(u)$ stored in a C-ART. Red highlights the updated pointers.}
	\label{fig:c-art-insert}
\end{figure}

4) \textit{Delete(u, v)} removes vertex $v$ from $N(u)$. The operation begins by locating $v$ using \textit{Search(u, v)} and then removing it from the leaf. If the number of vertices in the leaf drops below $\frac{B}{2}$ after deletion, the leaf is checked for potential merging with adjacent sibling leaves to maintain a high filling ratio.


5) When two neighbor sets are stored in C-ARTs, set intersections can be performed efficiently based on the structures of the two trees.

Note that both insert and delete operations are executed on the copied root-to-leaf paths to ensure consistency and isolation.

\vspace{2pt}
\noindent\textbf{Optimization.} First, vertex IDs are compressed by removing the LCP and storing only the unique suffixes. During computation, the full vertex ID is reconstructed by concatenating the LCP with the stored suffix. This approach reduces memory usage and minimizes data movement between memory and the CPU, accelerating computation. If the vertex IDs in a leaf differ only in their last byte, they are stored using a 256-bit bitmap. Second, to address the overhead of iterating through up to 256 child pointers per node, many of which may be empty, a bitmap is maintained in each node to record the presence of non-empty pointers. By utilizing AVX2 instructions, we efficiently identify the indices of set bits in the bitmap, bypassing empty pointers and significantly improving performance.

\subsection{Clustered Index} \label{sec:clustered_index}

The overhead of storing neighbor sets for low-degree vertices in C-ART can diminish its benefits, as the tree's depth depends on the vertex ID length. Therefore, we introduce the \emph{clustered index} to store neighbor sets of low-degree vertices. The clustered index is implemented as a B+ tree where the keys are edge pairs $e(u, v)$ representing source and destination vertices. Neighbor sets of low-degree vertices in $V(S)$ are stored sequentially within the clustered index according to the $(u, v)$ order, eliminating the random memory access overhead caused by traversing different neighbor sets. The vertex array records the position of each neighbor set within the clustered index to accelerate the visit of $N(u)$. Since $|V(S)|$ is small (64 by default) and the clustered index exclusively stores neighbor sets of low-degree vertices, the depth of the clustered index remains low, typically two levels. Consequently, graph access operations and updates are executed efficiently.

\subsection{Garbage Collection of Graph Store}

RapidStore employs a classical reference counting method for GC. Specifically, each node and leaf maintains a reference count that records the number of parents from different snapshots referring to it. The reference count is incremented by one during the copy-on-write operation along the root-to-leaf path. Conversely, when graph concurrency control claims a snapshot version, RapidStore traverses from top to bottom, decrementing the reference count. If a node’s reference count drops to zero, it is reclaimed. 
\begin{table*}[t]
\centering
\setlength{\abovecaptionskip}{0pt}
\setlength{\belowcaptionskip}{0pt}
\caption{Performance of graph analytics. For the CSR, we report the latency time. For the systems, we report the slowdown over CSR. The best performance is in bold. "-" means not supported. "OOT" means it cannot be completed within 5 hours.}
\label{tab:graph_analytics}

\resizebox{\textwidth}{!}{
\begin{tabular}{|c|ccccc|ccccc|ccccc|}
\hline
\multirow{2}{*}{\textbf{}} & \multicolumn{5}{c|}{\textit{\textbf{lj}}} & \multicolumn{5}{c|}{\textit{\textbf{ot}}} & \multicolumn{5}{c|}{\textit{\textbf{ldbc}}} \\ \cline{2-16} 
 & \textbf{BFS} & \textbf{PR} & \textbf{SSSP} & \textbf{WCC} & \textbf{TC} & \textbf{BFS} & \textbf{PR} & \textbf{SSSP} & \textbf{WCC} & \textbf{TC} & \textbf{BFS} & \textbf{PR} & \textbf{SSSP} & \textbf{WCC} & \textbf{TC} \\ \hline
CSR & 0.56s & 4.66s & 1.42s & 1.95s & 48.21s & 0.64s & 7.91s & 1.64s & 1.73s & 232.80s & 2.68s & 20.57s & 6.36s & 5.73s & 442.10s \\ \hline
Sortledton & 7.23 & 6.69 & 2.45 & 3.11 & 3.18 & 5.97 & 4.77 & 2.46 & 3.07 & 3.07 & 8.35 & 8.49 & 3.04 & 4.20 & 4.94 \\
Teseo & 15.25 & 4.32 & 3.53 & 2.52 & 3.99 & 3.83 & 2.74 & 3.50 & 4.29 & 4.34 & 6.28 & 4.01 & 3.84 & 3.23 & 7.31 \\
Aspen & 7.11 & 7.10 & 4.81 & 3.97 & 6.24 & 3.67 & 3.34 & 2.71 & 2.80 & 10.15 & 8.90 & 9.07 & 5.97 & 4.58 & 9.87 \\
LiveGraph & 21.77 & 17.12 & 4.76 & 7.01 & - & 15.77 & 10.49 & 7.84 & 6.86 & - & 11.46 & 13.91 & 4.99 & 6.41 & - \\
\textbf{RapidStore} & \textbf{2.11} & \textbf{1.53} & \textbf{1.61} & \textbf{1.18} & \textbf{1.39} & \textbf{1.59} & \textbf{1.35} & \textbf{1.28} & \textbf{1.22} & \textbf{1.51} & \textbf{1.82} & \textbf{1.65} & \textbf{1.79} & \textbf{1.36} & \textbf{1.49} \\ \hline
\multirow{2}{*}{} & \multicolumn{5}{c|}{\textit{\textbf{g5}}} & \multicolumn{5}{c|}{\textit{\textbf{tw}}} & \multicolumn{5}{c|}{\textit{\textbf{fr}}} \\ \cline{2-16} 
 & \textbf{BFS} & \textbf{PR} & \textbf{SSSP} & \textbf{WCC} & \textbf{TC} & \textbf{BFS} & \textbf{PR} & \textbf{SSSP} & \textbf{WCC} & \textbf{TC} & \textbf{BFS} & \textbf{PR} & \textbf{SSSP} & \textbf{WCC} & \textbf{TC} \\ \hline
CSR & 2.00s & 20.98s & 4.25s & 4.88s & 5268.23s & 2.92s & 24.87s & 6.71s & 10.48s & 4631.08s & 31.80s & 323.60s & 68.20s & 65.39s & 1974.24s \\ \hline
Sortledton & 4.42 & 4.49 & 2.42 & 3.06 & 3.57 & 6.70 & 6.80 & 2.60 & 3.45 & 3.47 & 3.18 & 3.79 & 1.92 & 3.58 & 9.62 \\
Teseo & 4.65 & 2.37 & 3.02 & 2.70 & OOT & 4.61 & 3.36 & 3.87 & 2.48 & OOT & 1.91 & 2.10 & 2.13 & 2.04 & 10.35 \\
Aspen & 4.35 & 4.21 & 3.45 & 2.93 & OOT & 7.61 & 7.03 & 5.23 & 3.89 & OOT & 2.41 & 1.80 & 2.81 & 2.32 & OOT \\
LiveGraph & 8.90 & 9.64 & 4.07 & 5.74 & - & 12.39 & 12.51 & 4.46 & 6.29 & - & OOM & OOM & OOM & OOM & - \\
\textbf{RapidStore} & \textbf{1.60} & \textbf{1.61} & \textbf{1.60} & \textbf{1.24} & \textbf{1.89} & \textbf{1.71} & \textbf{1.68} & \textbf{1.80} & \textbf{1.33} & \textbf{2.26} & \textbf{0.92} & \textbf{1.04} & \textbf{1.17} & \textbf{1.16} & \textbf{3.67} \\ \hline

\end{tabular}
}
\end{table*}
\subsection{Analysis of Multi-Version Graph Data Store}

The clustered index follows the classical B+ tree's time and space complexities. Here, we analyze the time and space costs of C-ART. Let $|P|$ denote the partition size, $B$ the leaf segment size, $d = |N(u)|$, and $w$ the length of a vertex ID in bytes.

\vspace{2pt}
\noindent\textbf{Time Cost.} Locating $N(u)$ for a vertex $u$ takes $O(1)$ time using the vertex index, which can be omitted. Compressing leaves does not increase the number of nodes compared to ART. Consequently, \textit{Scan(u)} operates in $O(d)$ time, the same as ART, while \textit{Search(u, v)} requires $O(w + \log B)$ time due to the binary search within the leaf. For \textit{Insert(u, v)}, two operations are involved:
1) Copying the root-to-leaf path, which takes $O(wB)$ time since the depth of C-ART is bounded by $w$.
2) Inserting $v$, which requires $O(|P| + w + \log B)$ time to find the insertion position and potentially split the leaf. Because of the copy-on-write strategy, delete has the same cost as insert.  
For set intersections between two neighbor sets $N(u)$ and $N(v)$ with degrees $d_1$ and $d_2$ (assume $d_1 \leqslant d_2$), two strategies are employed:
If $\frac{d_2}{d_1} < 10$, a merge-based set intersection is performed, taking $O(d_1 + d_2)$ time. Otherwise, for each neighbor in $N(u)$, existence in $N(v)$ is checked, taking $O(d_1 \times (w + \log B))$ time.

\sun{For vertex operations, \emph{VertexDelete(u)} removes vertex $u$ by first deleting all its incident edges $e(u, v)$ for each $v \in N(u)$. This involves acquiring locks on the subgraphs containing the neighbors, following MV2PL rules. After edge deletion, a flag bit is unset to mark the vertex as removed, and its ID is added to a queue for potential reuse. Vertex deletions are generally rare. \emph{VertexInsert(u)} adds vertex $u$ by first checking the queue for reusable IDs. If available, one is reused; otherwise, $N$ is atomically incremented to assign a new ID. The vertex is then added to the appropriate subgraph by setting a flag bit, requiring a lock only on that subgraph. The time complexity of \textit{VertexDelete(u)} is proportional to the number of adjacent edges to be deleted. After edge removal, deleting the vertex from the subgraph incurs an additional $O(|P|)$ cost due to snapshot creation. \textit{VertexInsert(u)} has the same $O(|P|)$ cost, as it involves no additional operations.} 

\sun{In summary, RapidStore can incur higher write overhead than per-edge versioning approaches due to its copy-on-write strategy. This design trade-off is intentional to optimize read performance, as graph workloads are typically read-intensive. Nevertheless, write operations remain efficient in practice, benefiting from the small values of $w$ and $B$, and experimental results confirm that RapidStore achieves good write performance.}



\vspace{2pt}
\noindent\textbf{Space Cost.} Given $S$, the vertex index requires $O(|P|)$ space. RapidStore uses $O(d)$ space to store $N(u)$ since each leaf contains at least one vertex. Each entry in the clustered index corresponds to a single neighbor. Thus, the overall space complexity for storing the graph is $O(|V(G)| + |E(G)|)$. The additional overhead from multi-versioning is minimal. The number of subgraph versions is bounded by the number of concurrent read queries (as detailed in Section~\ref{sec:concurrency_control}), and each snapshot only duplicates a root-to-leaf path, which incurs negligible space overhead due to its small size.

\vspace{2pt}
\noindent\textbf{Hyperparameters.} RapidStore has two hyperparameters: the partition size $|P|$ and the segment size $B$. As analyzed above, increasing these values can improve read efficiency since $N(u)$ are stored in larger chunks. However, larger values decrease write performance. We empirically set $|P|$ and $B$ to 64 and 512 to balance read and write performance, and this value cooperates with AVX2 instructions.



%% file: 7_experiment.tex
\section{Experiments} \label{sec:experiments}


\textbf{Experiment Setup.} We conduct our experiments on a machine equipped with Intel Xeon Gold 6430 @ 3.40GHz processors. The machine features 256GB of DDR5 memory and a maximum bandwidth of 61.2GB/s. The CPU has a 60MB L3 Cache and 32 cores. We compile the source code using GCC v10.5.0 with O3 optimization. Each experiment is executed five times, and we report the median.

\vspace{2pt}
\noindent\textbf{Graph Datasets. }  We utilize a diverse set of graph datasets to evaluate the systems' time and space performance, as depicted in Table \ref{tab:datasets}. These graphs vary in size and structure and are widely used in previous graph research, enhancing the comprehensiveness of our results. Table \ref{tab:datasets} provides a brief description of each dataset.

\begin{table}[t]
\setlength{\abovecaptionskip}{0pt}
\setlength{\belowcaptionskip}{0pt}
\captionsetup{skip=0pt} 
\centering
\caption{The detailed information of the graph datasets.}
\label{tab:datasets}
\resizebox{0.45\textwidth}{!}{
\begin{tabular}{|c|c|c|c|c|c|c|}
\hline
\textbf{Dataset} & \textbf{Abbr.} & \textbf{$|V|$} & \textbf{$|E|$} & \textbf{Avg. Deg.} & \textbf{Max Deg.} & \textbf{Size(GB)} \\ \hline
LiveJournal & \textit{lj} & 4M & 43M & 17.4 & 14815 & 0.67 \\
Orkut & \textit{ot} & 3M & 117M & 76.3 & 33313 & 1.7 \\
LDBC & \textit{ldbc} & 30M & 176M & 11.8 & 4282595 & 2.84 \\
Graph500 & \textit{g5} & 9M & 260M & 58.7 & 406416 & 4.16 \\
Twitter & \textit{tw} & 21M & 265M & 24.8 & 698112 & 4.11 \\ \
Friendster & \textit{fr} & 65M & 2B & 55.1 & 5214 & 30.1\\ \hline
\end{tabular}
}
\end{table}

\vspace{2pt}
\noindent\textbf{Graph Analytic Workload. } For graph analytic workloads, we select five algorithms from GAPBS~\cite{beamer2015gap}: PageRank (PR), breadth-first search (BFS), single-source shortest path (SSSP), weakly connected components (WCC), and Triangle Counting (TC). These workloads cover a range of graph data access patterns and represent common graph analytics tasks. The parameters for each algorithm were set according to standard practices, such as 10 iterations for PageRank.

\vspace{2pt}
\noindent\textbf{Systems Under Study.}
We compare RapidStore with several of the latest graph systems: \textbf{Sortledton} \cite{fuchs2022sortledton} is a library using a two-level array to store vertices. It stores small neighbors directly in arrays, while large neighbors are stored in unrolled skip lists. \textbf{Teseo} \cite{de2021teseo} is a library that stores vertices and edges together in a Packed Memory Array (PMA)~\cite{de2019packed} indexed by ART. \textbf{Aspen}~\cite{dhulipala2019low} applies copy-on-write and versions the global graph, focusing on read performance. It uses PAM tree~\cite{sun2018pam} to store edges. \textbf{LiveGraph}~\cite{zhu2019livegraph} stores the vertex's neighbors as logs in timestamp order, improving insertion and \emph{scan} but limiting its functionality. \textbf{Spruce}~\cite{shi2024spruce} is a library that uses an ART-like structure to index vertices and buffer block and sorted arrays to store neighbors independently. However, Spruce does not provide full isolation support and faces OOM issues on the three larger datasets. Therefore, we did not include Spruce. We also consider \textbf{LSMGraph}~\cite{yu2024lsmgraph}, leveraging an LSM tree to optimize disk performance. However, the source code is currently unavailable.  All systems are implemented in C++. Additionally, we include CSR as a baseline to demonstrate optimal static performance.

\vspace{2pt}
\noindent\textbf{Supplement Experiments.} \sun{Due to space constraints, we present additional results including basic read operations, multicore write scalability, batch updates, real insert traces on \emph{ldbc}, and insertions with varying neighbor sizes in the full version of the paper.}

\subsection{Evaluation on Read Performance}

Table \ref{tab:graph_analytics} shows the performance of graph analytics workloads.  \textbf{PR} and \textbf{WCC}, which involve sequential vertex and neighborhood access, demonstrate significant improvements with RapidStore. It reduces latency by 31.86-64.50\% for PR and 43.30-57.91\% for WCC compared to the best-performing system, showcasing its superior scan performance. Teseo ranks second, benefiting from its PMA's locality advantages. \textbf{BFS}, characterized by random vertex access and sequential neighborhood access, sees a latency reduction of 51.71-71.08\% with RapidStore. For \textbf{SSSP}, which requires random neighborhood access, RapidStore reduces latency by 30.81-47.89\%. Sortledton ranks second on most datasets, with other systems showing similar performance trends. These results highlight RapidStore's capability to efficiently handle random access patterns, a critical requirement in dynamic graph applications.

\textbf{TC} requires efficient \textit{intersect} operations, which depend on fast \textit{search} and \textit{scan}. RapidStore achieves latency reductions of 34.85-69.84\%, demonstrating robust performance in these tasks. Notably, RapidStore exhibits balanced performance across workloads, occasionally showing smaller slowdowns on sequential \textit{Scan}-based tasks like PR. This indicates its adaptability to diverse access patterns, making it highly versatile for graph analytics workloads. In summary, RapidStore achieves significantly better read performance than existing systems.




\subsection{Evaluation on Write Performance} \label{sec:write_exp}

\begin{figure}[t]
	\setlength{\abovecaptionskip}{0pt}
	\setlength{\belowcaptionskip}{0pt}
		\captionsetup[subfigure]{aboveskip=0pt,belowskip=0pt}
    \includegraphics[width=0.45\textwidth]{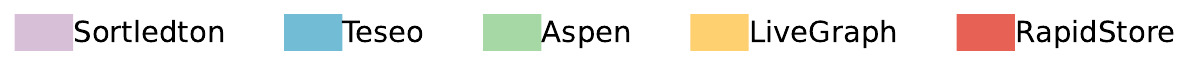}
	\begin{subfigure}[h]{0.45\textwidth}
		\centering
		\includegraphics[width=\textwidth]{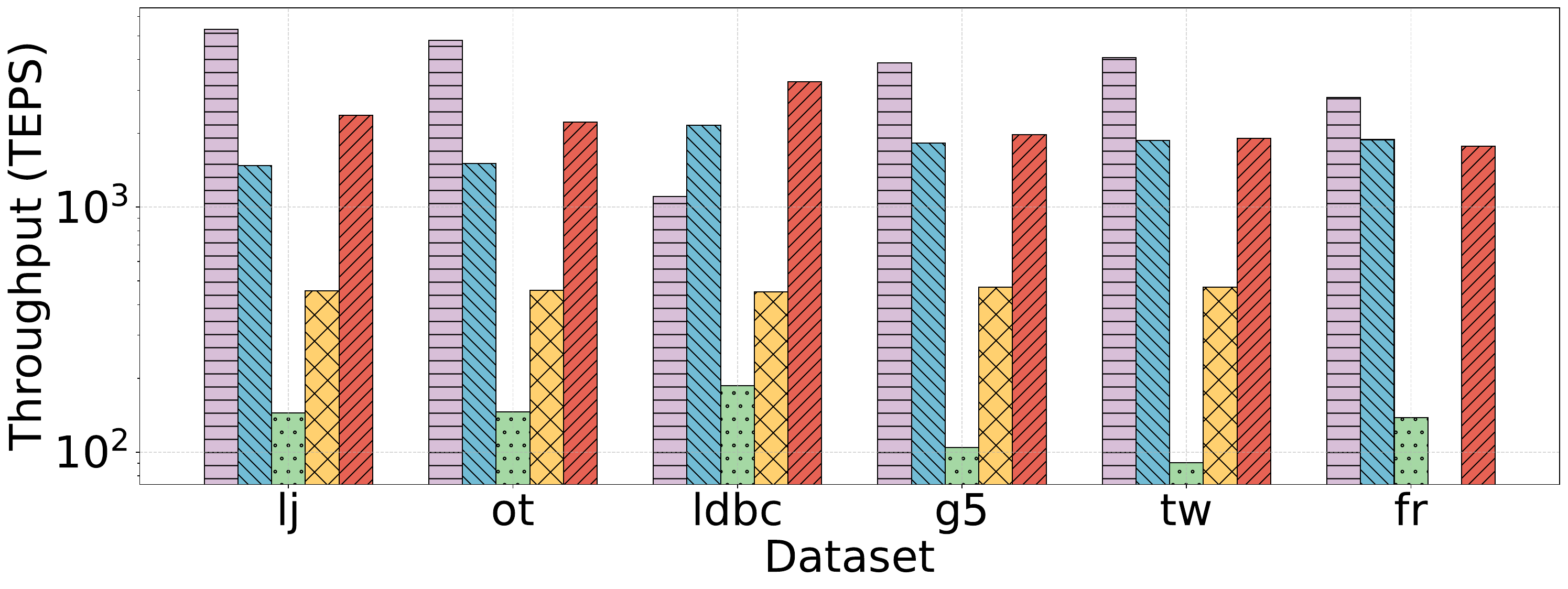}
	   \caption{Performance of edge insertion.}
	\label{fig:insert}
	\end{subfigure}
    \begin{subfigure}[h]{0.45\textwidth}
		\centering
		\includegraphics[width=\textwidth]{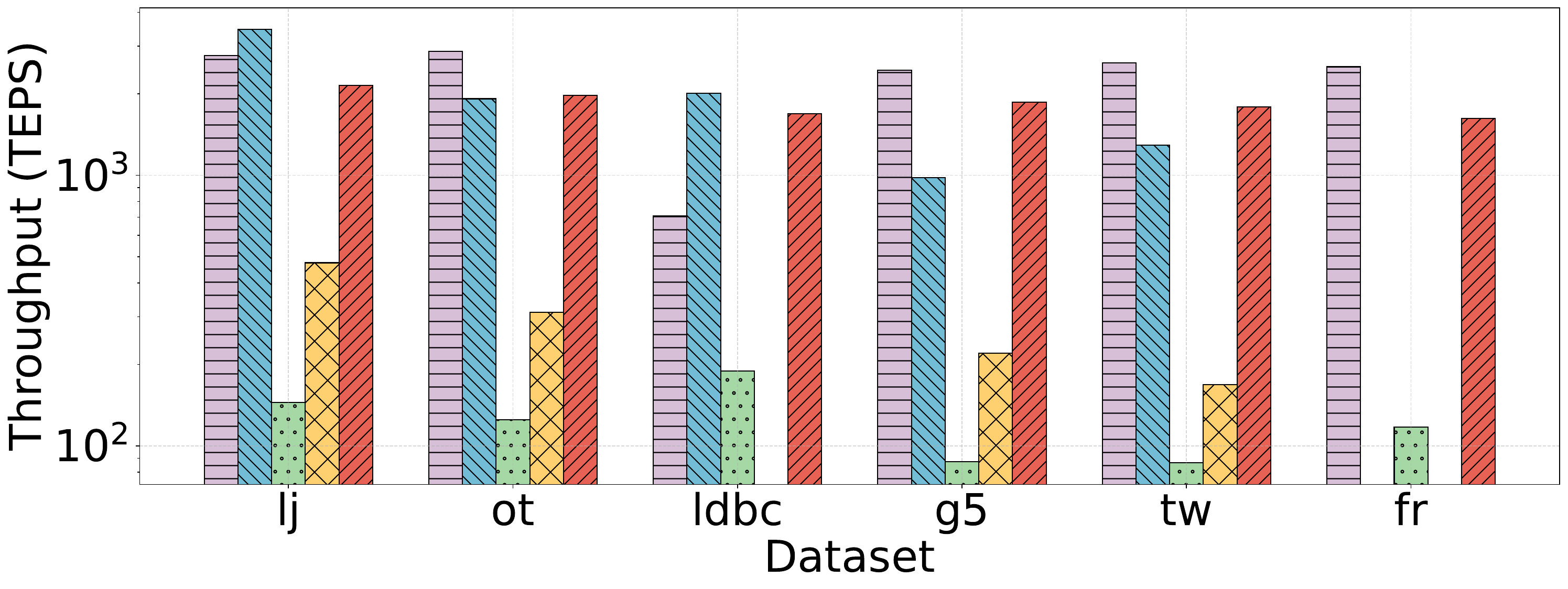}
        \caption{Performance of edge update.}
        \label{fig:update}
	\end{subfigure}
	\caption{Performance of write operations.}
\end{figure}
In this section, we evaluate the system's edge insertion and deletion performance. While RapidStore supports property updates, this was excluded from the evaluation because Sortledton crashes during updates, and Teseo and Aspen do not support this functionality.

\vspace{2pt}
\noindent\textbf{Insert.} We evaluate insertion performance by inserting edges in random order using 32 writer threads. For comparison, Aspen, designed for single-writer execution, is also tested with 32 threads. The results are shown in Figure~\ref{fig:insert}. Sortledton achieves the highest throughput, outperforming all other systems. RapidStore ranks second, with slightly lower throughput due to the overhead of copy-on-write compared to direct edge insertion. However, it benefits from optimized \emph{search} and ID compression, which accelerate insertion preparation and reduce copy cost, resulting in only a 1.9–2.2x slowdown relative to Sortledton. On the \emph{ldbc} dataset, RapidStore achieves the best performance, as Sortledton suffers from severe lock contention caused by high skewness. Teseo ranks third, with throughput occasionally interrupted by PMA rebalancing. LiveGraph shows the lowest throughput, limited by inefficient \emph{search} operations. Overall, RapidStore effectively reduces insertion overhead through constant-time \emph{search} and compact ID encoding, particularly for high-degree vertices.

\vspace{2pt}
\noindent\textbf{Update.} This experiment evaluates update performance by repeatedly deleting and re-inserting 20\% of edges over five rounds, generating version chains for the modified elements. Figure~\ref{fig:update} shows the results. Sortledton’s throughput drops by 34.01\% compared to pure insertion, due to the overhead of managing version chains. Teseo performs well on small datasets but degrades on larger ones, where frequent updates trigger costly background garbage collection. In contrast, RapidStore and the other two systems show minimal performance change. RapidStore’s throughput drops by only 14.67\%, benefiting from constant-time \emph{search}. Overall, RapidStore delivers stable and competitive write performance, slightly behind Sortledton except on \emph{ldbc}, where Sortledton suffers from severe lock contention. These results highlight the sensitivity of Sortledton and Teseo to GC and locking, whereas RapidStore maintains consistent performance despite copy-on-write overhead.

\subsection{Evaluation on Concurrent Read and Write} \label{sec:concurrent_exp} 

This section evaluates concurrent read and write performance, analyzing the impact of writers on readers and vice versa. The goal is to assess the concurrency capabilities of the systems under study. To reduce the influence of memory bandwidth limitations, property storage is disabled for all systems. Aspen is excluded from this evaluation as it supports only a single writer.

\begin{figure}[t]
	\setlength{\abovecaptionskip}{0pt}
	\setlength{\belowcaptionskip}{0pt}
		\captionsetup[subfigure]{aboveskip=0pt,belowskip=0pt}
	\centering
    \includegraphics[width=0.37\textwidth]{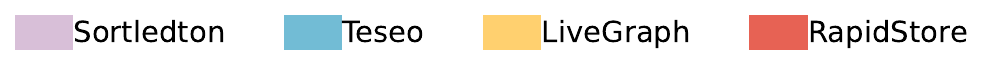}
    \\
	\begin{subfigure}[t]{0.45\textwidth}
		\centering
		\includegraphics[width=\textwidth]{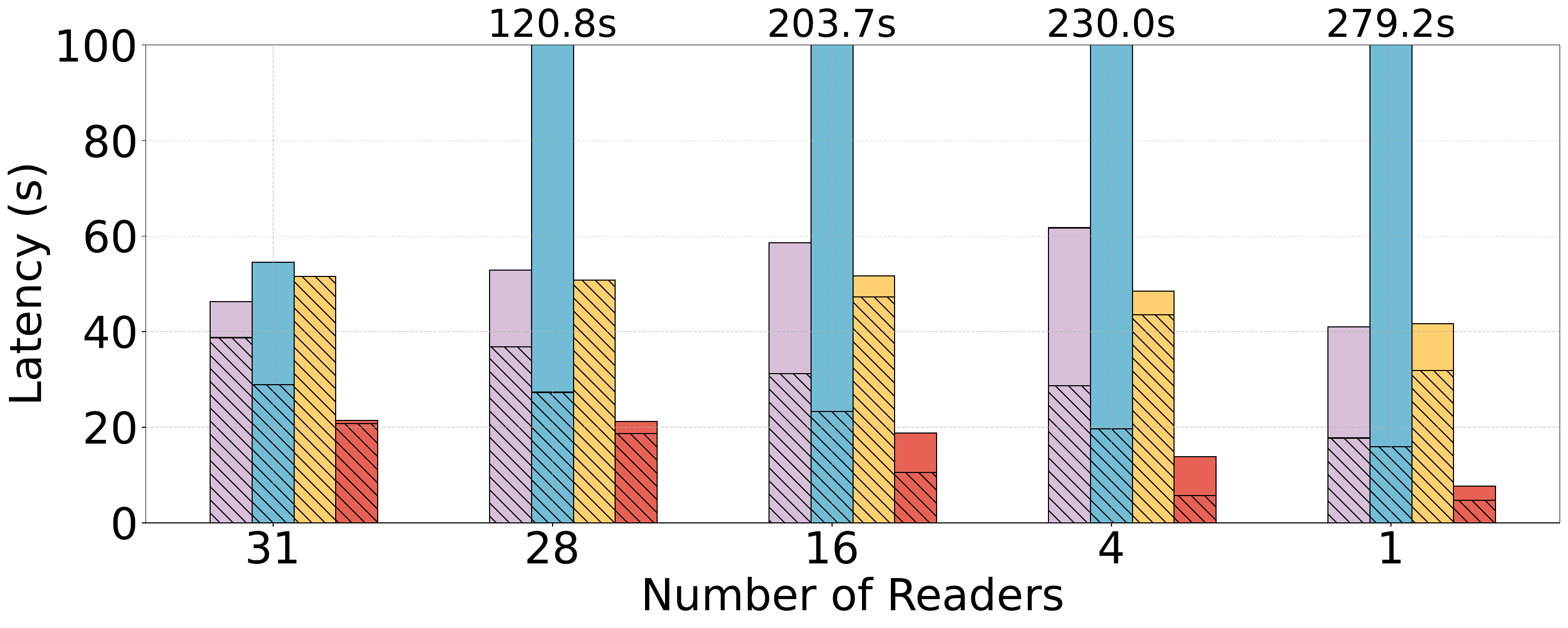}
		\caption{PR on \textit{lj}.}
		\label{fig:concurrent_read_lj}
	\end{subfigure}
        \begin{subfigure}[t]{0.45\textwidth}
		\centering
		\includegraphics[width=\textwidth]{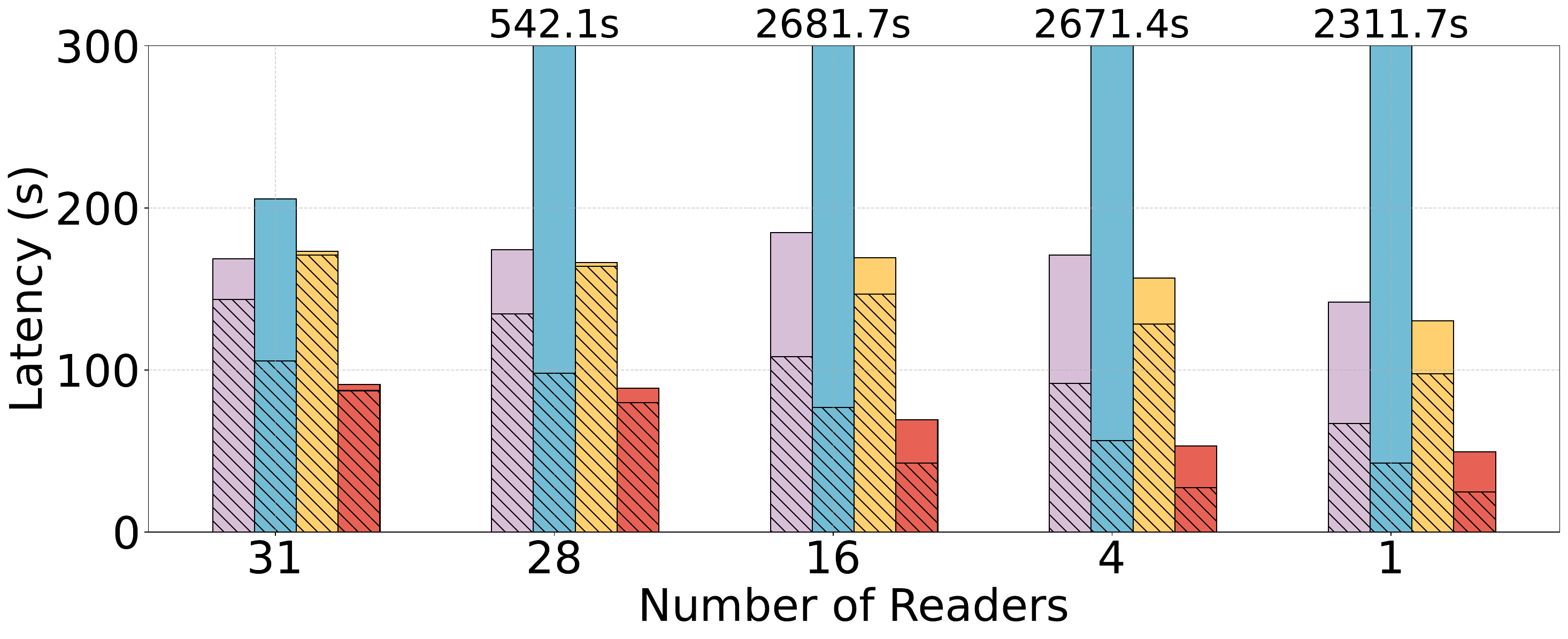}
		\caption{PR on \textit{g5}.}
		\label{fig:concurrent_read_g5}
	\end{subfigure}
	\caption{Read performance under varying numbers of readers and writers (total threads fixed at 32). The shadowed bars represent the latency of readers in the absence of writers.}
	\label{fig:concurrent_read}
\end{figure}

\vspace{2pt}
\noindent\textbf{Read with Concurrent Writers.} We first load the graph and then run a mixed workload with 32 threads, combining PR readers and edge update writers (delete + re-insert) to keep the graph size stable. As shown in Figure~\ref{fig:concurrent_read}, RapidStore consistently achieves the lowest latency and best performance. While Teseo and Sortledton perform well under isolated workloads, their read performance degrades significantly under concurrent access due to severe lock contention between readers and writers. In contrast, RapidStore experiences minimal interference, showing negligible performance drop even with 31 or 28 writers.

Although RapidStore’s decoupled design avoids reader-writer blocking, its read performance declines when the number of writers increases to 28 or 31 (i.e., only 4 or 1 readers remain). This is due to high scan throughput saturating memory bandwidth. As shown in Figure~\ref{fig:concurrent_write_bandwidth}, adding writers introduces bandwidth contention, making RapidStore sensitive to memory bandwidth pressure. In summary, Teseo and Sortledton are limited by lock contention, while RapidStore is constrained by hardware bandwidth, demonstrating its ability to fully utilize system resources. It’s worth noting that such extreme write-heavy workloads are rare in practice, as: 1) graph applications are typically read-intensive, and 2) adding more writers yields diminishing returns due to increased contention (see 
evaluation on Multicore Scalability in the full version).

Interestingly, LiveGraph shows little impact from concurrent writers. This is because its slow scan and insert operations lead to low bandwidth usage. Additionally, it stores neighbor sets in append order, allowing readers to record set lengths without locking. However, as shown in prior experiments, LiveGraph suffers from poor overall performance due to its unsorted neighbor sets.

\begin{figure}[t]
	\setlength{\abovecaptionskip}{0pt}
	\setlength{\belowcaptionskip}{-10pt}
		\captionsetup[subfigure]{aboveskip=0pt,belowskip=0pt}
	\centering
    \includegraphics[width=0.37\textwidth]{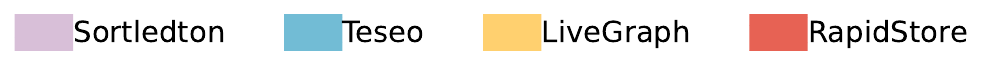}
    \\
    	\begin{subfigure}[h]{0.45\textwidth}
		\centering
		\includegraphics[width=\textwidth]{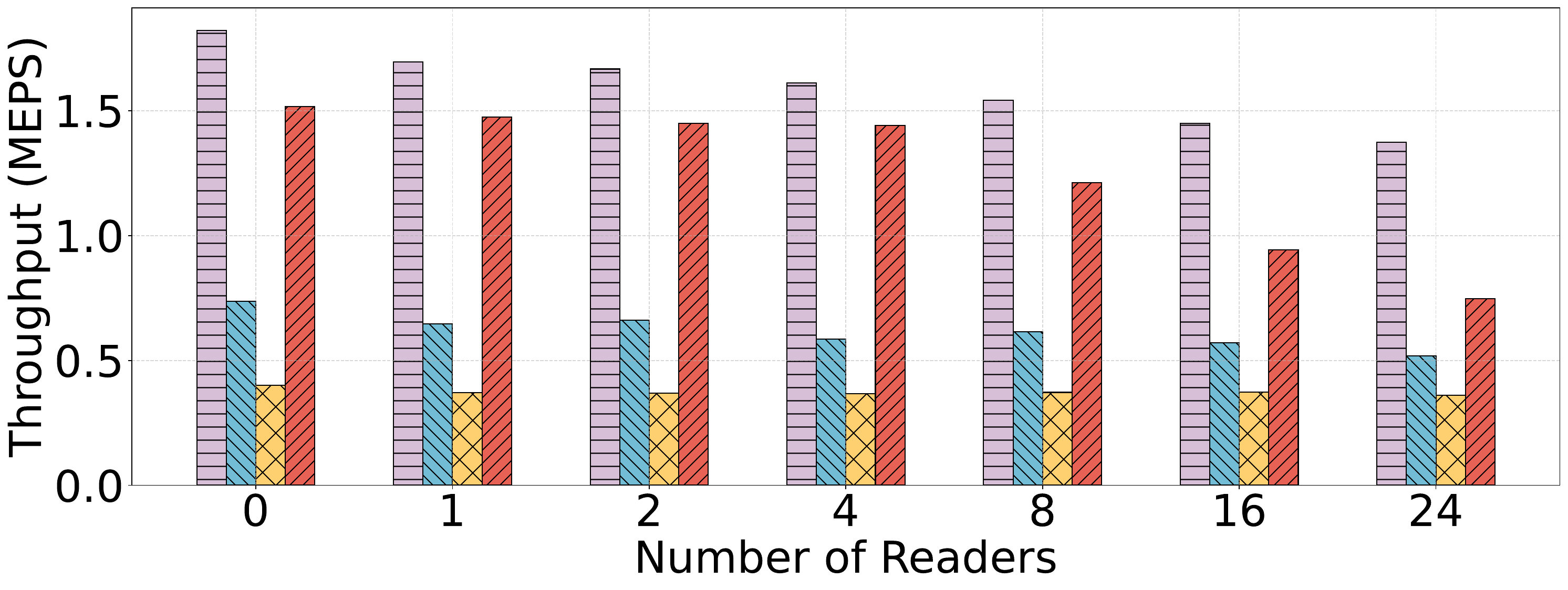}
		\caption{Insertion on \textit{lj}.}
		\label{fig:concurrent_write_lj}
	\end{subfigure}
        \begin{subfigure}[h]{0.45\textwidth}
		\centering
		\includegraphics[width=\textwidth]{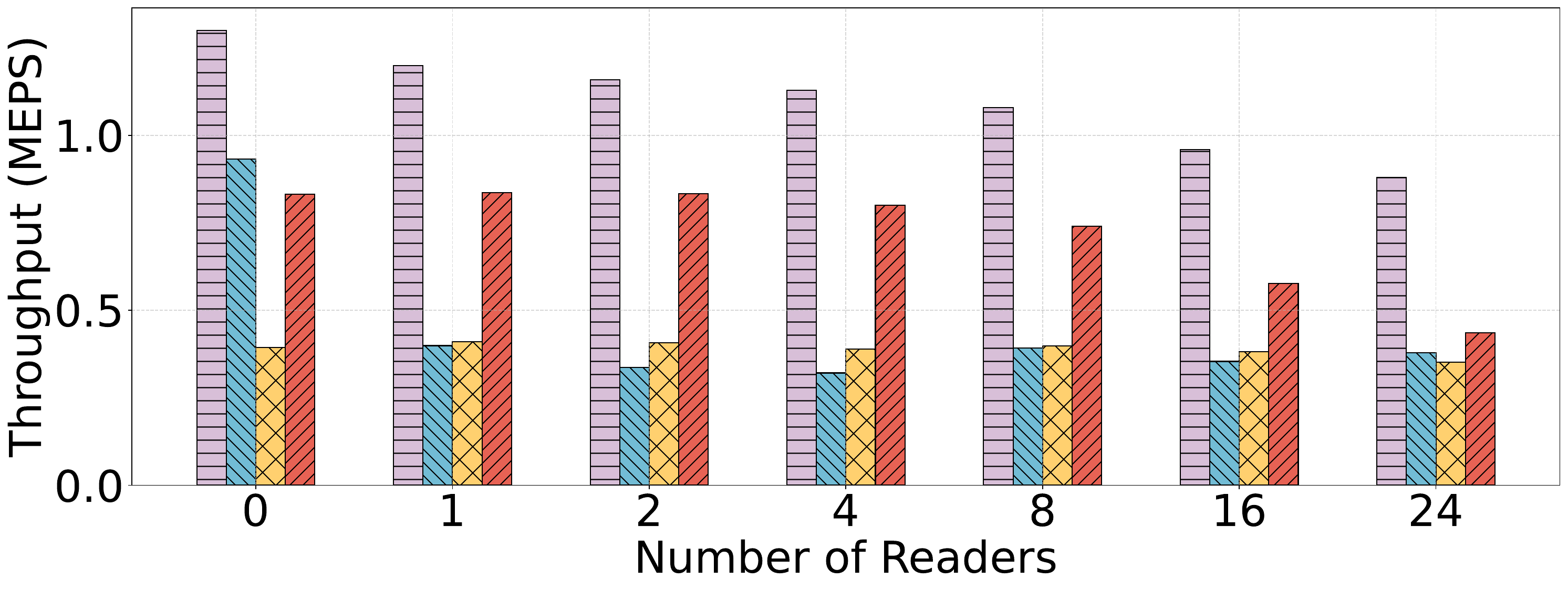}
		\caption{Insertion on \textit{g5}.}
		\label{fig:concurrent_write_g5}
	\end{subfigure}
	\caption{Insertion performance under varying numbers of readers. Each reader executes PR independently. The number of writers is fixed at 8.}
	\label{fig:concurrent_write}
\end{figure}

\begin{figure}[t]
	\setlength{\abovecaptionskip}{0pt}
	\setlength{\belowcaptionskip}{0pt}
	\setlength{\textfloatsep}{0pt}
		\captionsetup[subfigure]{aboveskip=0pt,belowskip=0pt}
	\centering
    \includegraphics[width=0.47\textwidth]{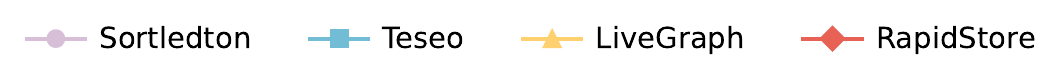}
    \\
    	\begin{subfigure}[h]{0.224\textwidth}
		\centering
		\includegraphics[width=\textwidth]{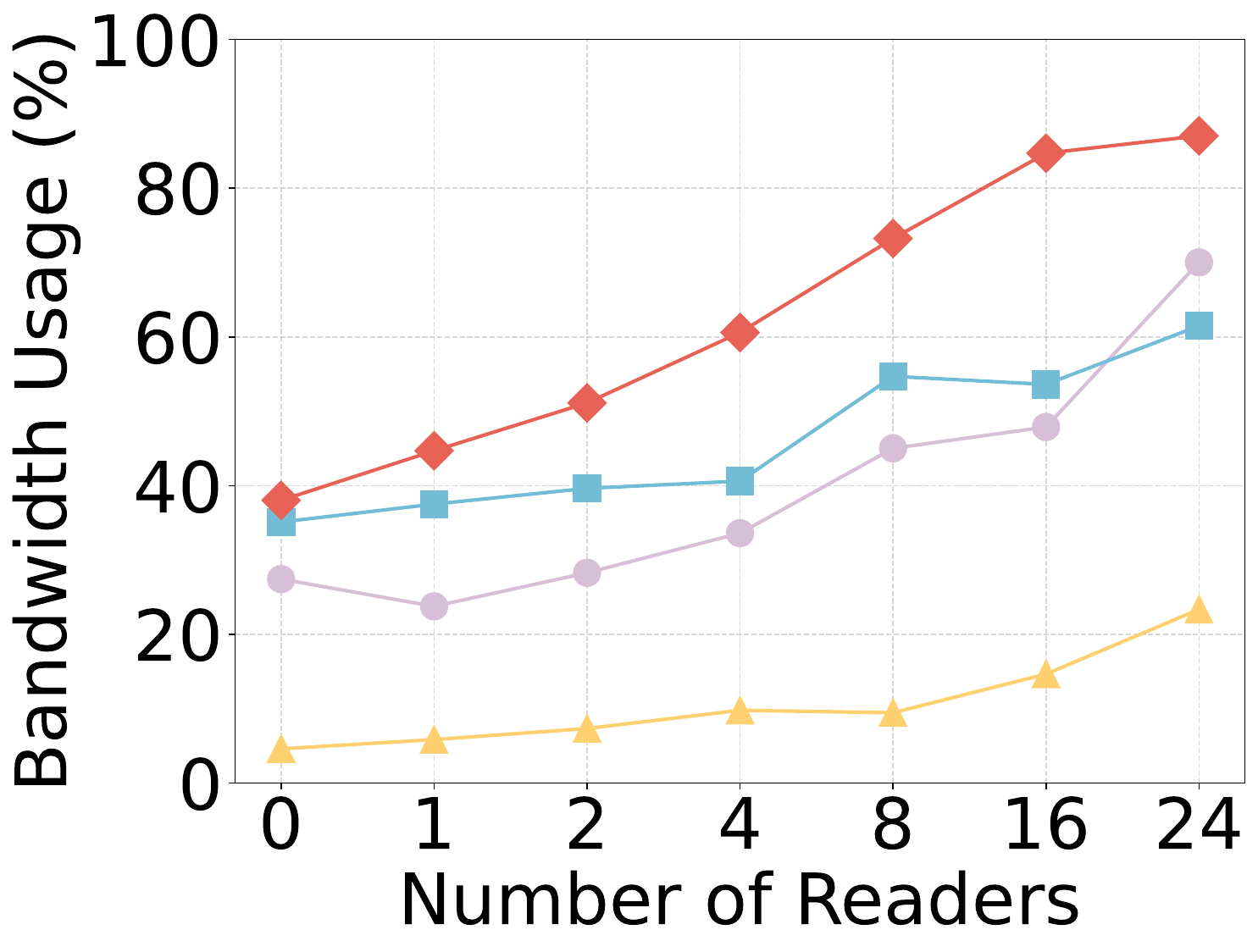}
		\caption{Bandwidth on \textit{lj}.}
		\label{fig:concurrent_write_bandwidth_lj}
	\end{subfigure}
        \begin{subfigure}[h]{0.23\textwidth}
		\centering
		\includegraphics[width=\textwidth]{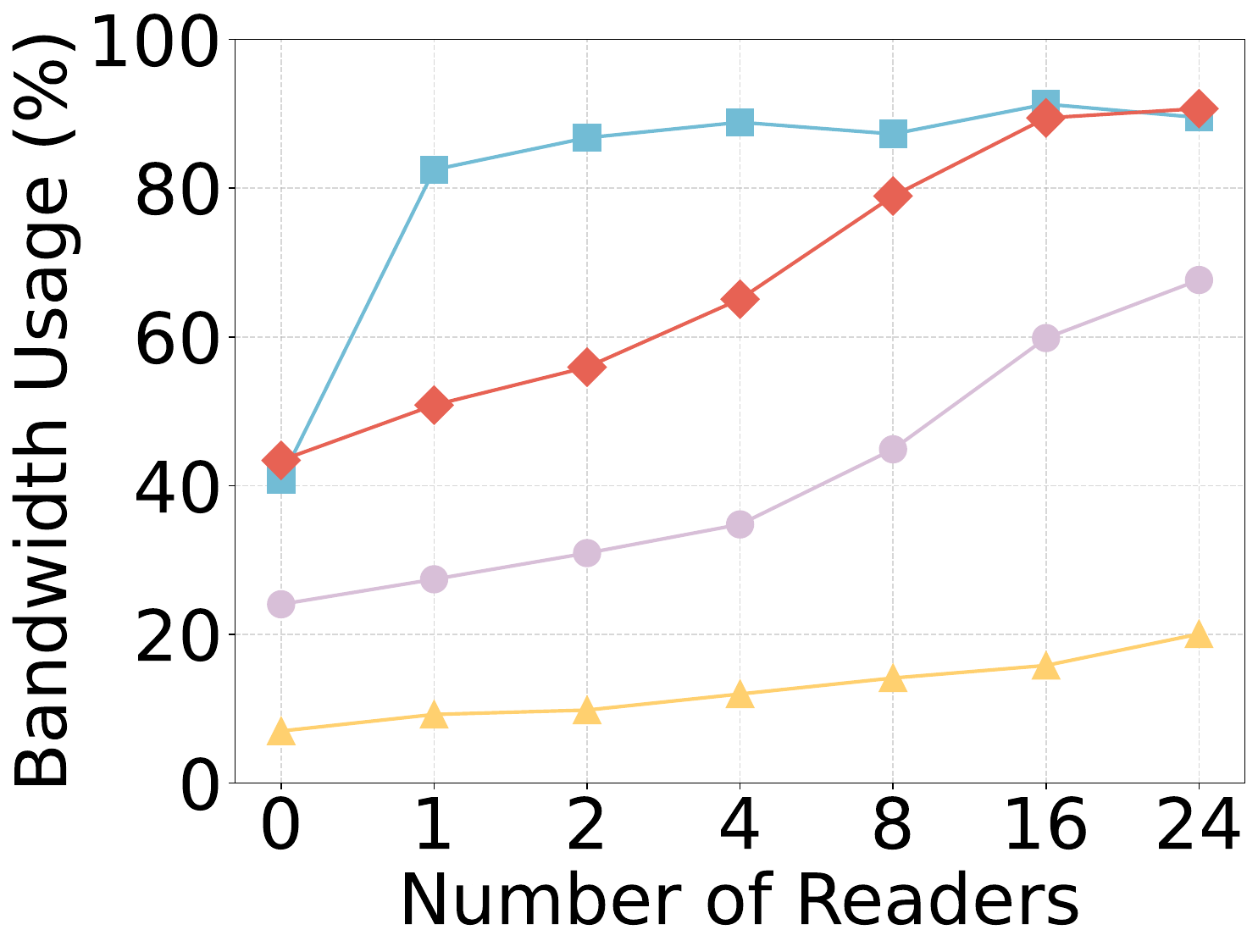}
		\caption{Bandwidth on \textit{g5}.}
		\label{fig:concurrent_write_bandwidth_g5}
	\end{subfigure}
	\caption{Corresponding avg. memory bandwidth usage rate in concurrent read and write experiments.}
	\label{fig:concurrent_write_bandwidth}
\end{figure}

\vspace{2pt}
\noindent\textbf{Write with Concurrent Readers.} We initially load 80\% of the edges and launch 8 writer threads. The number of readers varies from 0 to 24. Writers insert the remaining edges, while readers execute PageRank. Figures~\ref{fig:concurrent_write} and ~\ref{fig:concurrent_write_bandwidth} report insertion throughput and memory bandwidth utilization, respectively.

As the number of readers increases from 0 to 4, RapidStore’s insertion throughput drops by no more than 5.06\%, demonstrating strong resilience to reader interference. When \#readers exceeds 8, throughput declines further as RapidStore approaches the memory bandwidth limit due to its high scan performance. In contrast, Sortledton’s throughput drops by 13.29\% when \#readers increases from 0 to 4, despite low memory bandwidth usage. Teseo degrades more severely, with throughput reductions of 29.72\% and 57.22\% on the two datasets, primarily due to MVCC overhead. These results highlight the effectiveness of RapidStore’s decoupled design, which avoids reader-writer lock contention and fully utilizes available memory bandwidth. Its performance degradation is not due to contention, but rather due to saturating hardware limits.

\vspace{2pt}
\noindent\textbf{Summary.} RapidStore achieves high performance and efficient hardware utilization under concurrent read-write workloads. It delivers stable read and write performance in typical read-intensive configurations, such as 28 readers with 4 writers and 24 readers with 8 writers.



\subsection{Partition Size Evaluation} \label{sec:vertex_group_exp}

\begin{figure}[t]
	\setlength{\abovecaptionskip}{0pt}
	\setlength{\belowcaptionskip}{0pt}
		\captionsetup[subfigure]{aboveskip=0pt,belowskip=0pt}
	\centering
    \includegraphics[width=0.32\textwidth]{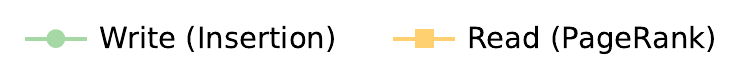}
    \\
	\begin{subfigure}[t]{0.23\textwidth}
		\centering
		\includegraphics[width=\textwidth]{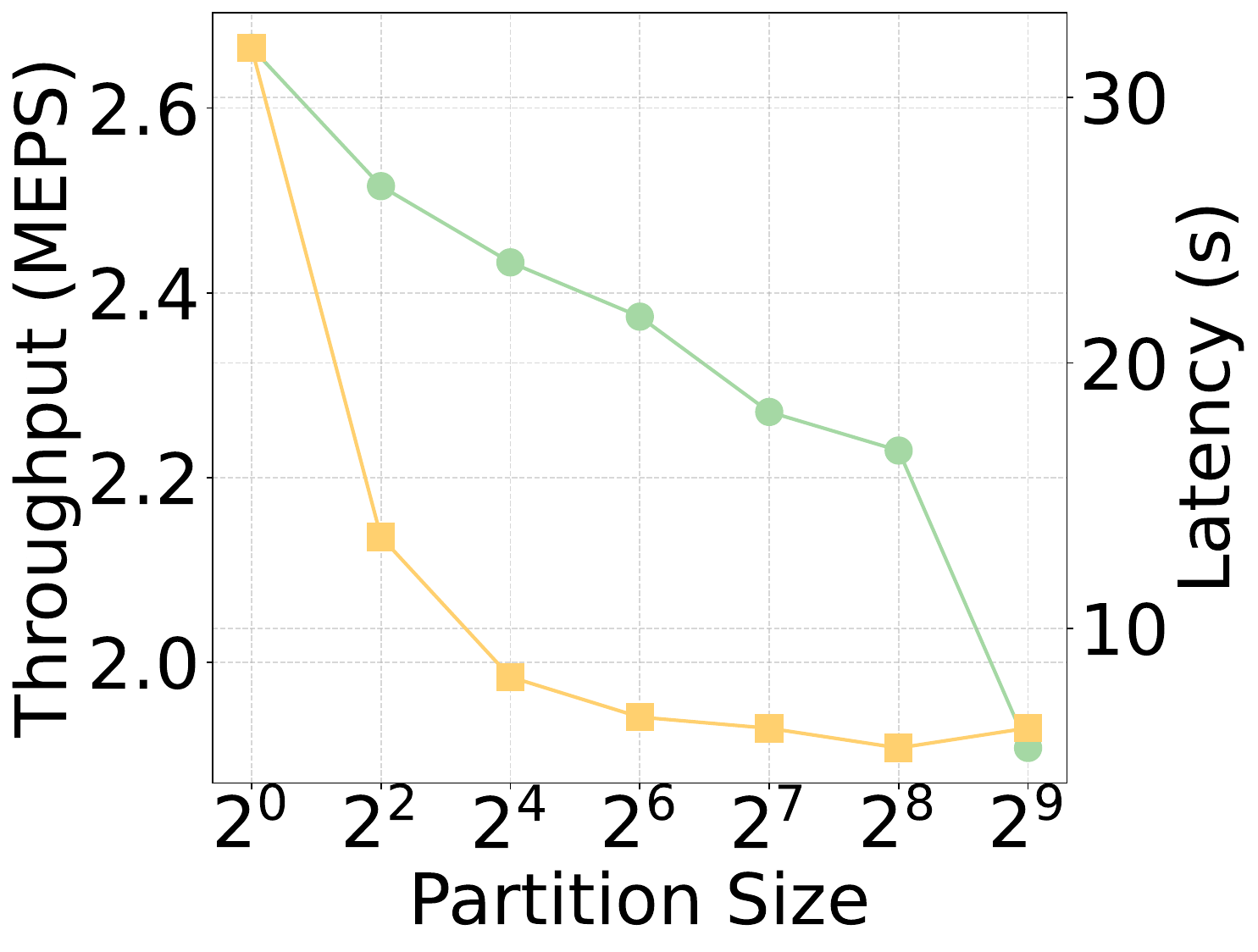}
		\caption{Performances on \textit{lj}.}
		\label{fig:group_size_lj}
	\end{subfigure}
        \begin{subfigure}[t]{0.23\textwidth}
		\centering
		\includegraphics[width=\textwidth]{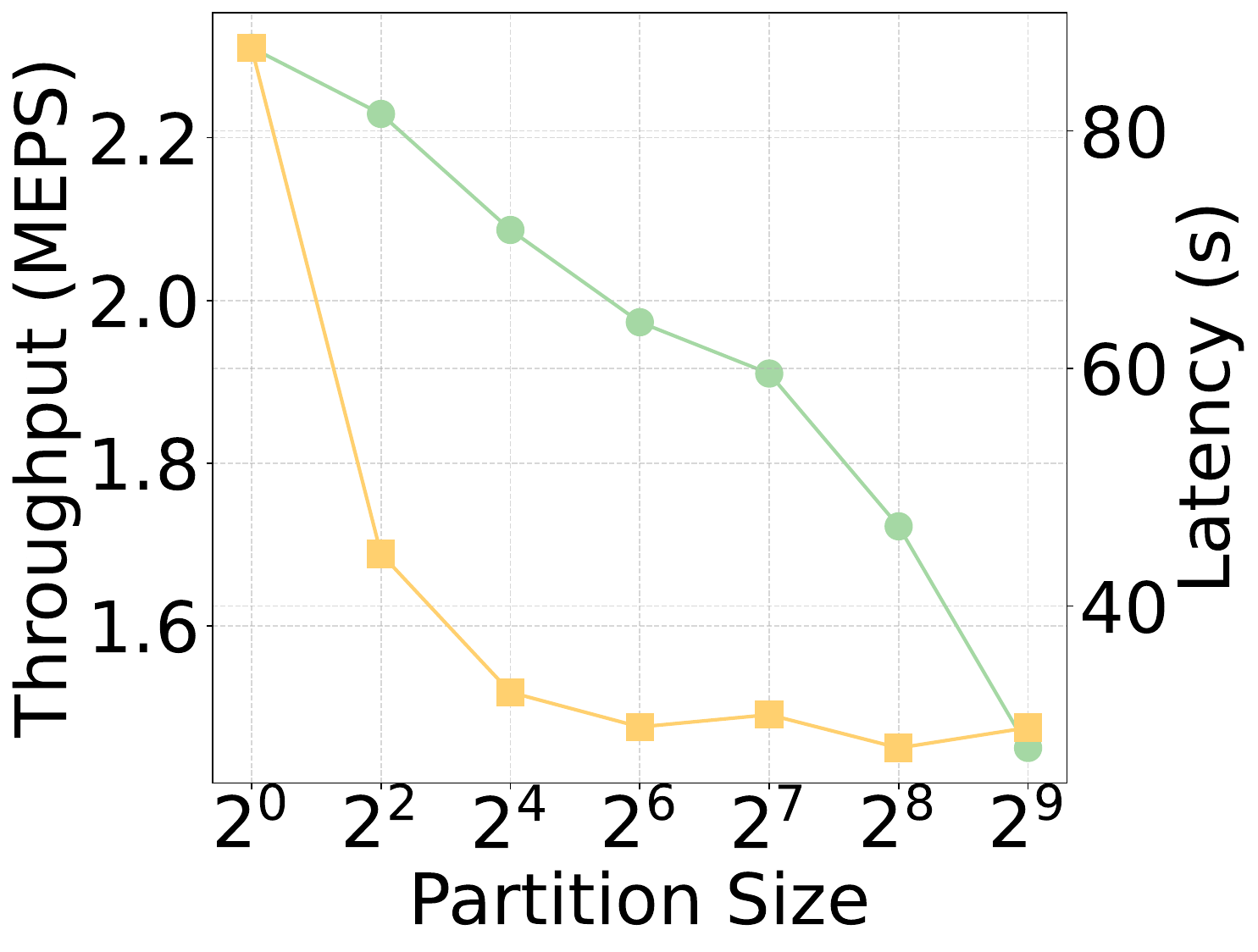}
		\caption{Performances on \textit{g5}.}
		\label{fig:group_size_g5}
	\end{subfigure}
	\caption{\sun{Write and read performance of RapidStore with varying partition sizes ($|P|$).}}
	\label{fig:group_size}
    \setlength{\textfloatsep}{0pt}
\end{figure}

\sun{We evaluate the impact of partition size $|P|$, which determines the snapshot granularity. To measure write and read performance, we use 32-writer insertion throughput (in million edges per second, MEPS) and PageRank (PR) latency, respectively. Figure~\ref{fig:group_size} shows the results. As $|P|$ increases, write performance decreases due to higher lock contention—each partition covers more vertices, increasing the chance of conflicts. In contrast, read performance improves initially, as more low-degree vertices are colocated, enhancing locality. However, beyond a certain point, PR latency plateaus or slightly increases, while insertion throughput continues to drop. These results confirm our analysis in Section~\ref{sec:concurrency_control}. Based on this, we set $|P| = 64$ by default, though users can tune it based on specific graph structures and workloads.}

\subsection{Ablation Study}

\begin{table}[t]
\setlength{\abovecaptionskip}{0pt}
\setlength{\belowcaptionskip}{0pt}
\captionsetup[subfigure]{aboveskip=0pt,belowskip=0pt}
\centering
\caption{\sun{Ablation study results. Insertion throughput is measured in thousands of edges per second (TEPS), and graph analytics latency is measured in seconds.}}
\resizebox{\columnwidth}{!}{
\begin{tabular}{|c|l|c|cccc|l|c|cccc|}
\hline
\multirow{2}{*}{\textbf{Dataset}} & \multicolumn{1}{c|}{\multirow{2}{*}{\textbf{Method}}} & \multirow{2}{*}{\textbf{Insert (TEPS)}} & \multicolumn{4}{c|}{\textbf{Analytics Query (s)}} \\ \cline{4-7} 
 & \multicolumn{1}{c|}{} &  & \multicolumn{1}{c|}{\textbf{BFS}} & \multicolumn{1}{c|}{\textbf{PR}} & \multicolumn{1}{c|}{\textbf{SSSP}} & \textbf{WCC} \\ \hline
\multirow{5}{*}{\textit{lj}} & ART & \textbf{2445.59} & \multicolumn{1}{c|}{10.92} & \multicolumn{1}{c|}{59.84} & \multicolumn{1}{c|}{11.01} & 17.04 \\
 & ART + SC & 2119.22 & \multicolumn{1}{c|}{7.31} & \multicolumn{1}{c|}{38.97} & \multicolumn{1}{c|}{5.42} & 11.55 \\
 & C-ART + SC & 1991.25 & \multicolumn{1}{c|}{7.30} & \multicolumn{1}{c|}{37.77} & \multicolumn{1}{c|}{5.25} & 10.81 \\
 & C-ART + SC + VEC & 1710.56 & \multicolumn{1}{c|}{2.42} & \multicolumn{1}{c|}{11.25} & \multicolumn{1}{c|}{2.73} & 2.59 \\
 & C-ART + SC + CI & 2373.05 & \multicolumn{1}{c|}{\textbf{1.05}} & \multicolumn{1}{c|}{\textbf{5.90}} & \multicolumn{1}{c|}{\textbf{2.31}} & \textbf{1.97} \\ \hline
\multirow{5}{*}{\textit{g5}} & ART & 1592.24 & \multicolumn{1}{c|}{25.22} & \multicolumn{1}{c|}{246.28} & \multicolumn{1}{c|}{32.21} & 45.98 \\
 & ART + SC & 1562.43 & \multicolumn{1}{c|}{22.85} & \multicolumn{1}{c|}{243.93} & \multicolumn{1}{c|}{27.46} & 41.86 \\
 & C-ART + SC & 1797.23 & \multicolumn{1}{c|}{9.73} & \multicolumn{1}{c|}{98.81} & \multicolumn{1}{c|}{13.20} & 18.87 \\
 & C-ART + SC + VEC & 1711.60 & \multicolumn{1}{c|}{4.18} & \multicolumn{1}{c|}{47.12} & \multicolumn{1}{c|}{8.23} & 9.06 \\
 & C-ART + SC + CI & \textbf{1973.13} & \multicolumn{1}{c|}{\textbf{2.71}} & \multicolumn{1}{c|}{\textbf{29.28}} & \multicolumn{1}{c|}{\textbf{6.47}} & \textbf{5.81} \\ \hline
\end{tabular}
}
\label{tab:ablation}
\end{table}

\sun{We conduct an ablation study to evaluate the performance impact of the three key techniques proposed in this paper: subgraph-centric concurrency control (denoted as SC), C-ART for the graph store, and the clustered index (denoted as CI) for managing low-degree vertices. As ART is a well-studied data structure and per-edge versioning is widely adopted in dynamic graph systems, we use ART combined with per-edge versioning as the baseline. In addition, since Sortledton optimizes performance by storing neighbor sets of low-degree vertices in separate vectors (denoted as VEC), we incorporate this technique into RapidStore for comparison with our clustered index design.}

\sun{Table~\ref{tab:ablation} presents the results. Enabling subgraph-centric concurrency control (SC) improves read performance by eliminating blocking on read queries caused by concurrent writes and per-edge versioning, while slightly affecting write performance due to coarser-grained versioning. This confirms both our analysis and the effectiveness of SC. Replacing ART with C-ART significantly improves analytics performance on \emph{g5}, while the improvement on \emph{lj} is limited. This is because \emph{lj} is a sparse graph where most vertices have small degrees, whereas \emph{g5} contains many high-degree vertices that pose greater performance challenges. Incorporating VEC improves read performance, demonstrating the benefit of using specialized structures for low-degree vertices. However, VEC degrades write performance due to the overhead of managing separated vectors. In contrast, our clustered index (CI) not only improves read performance significantly but also enhances write performance. These results validate the effectiveness of each individual technique.}

\subsection{Memory Consumption} \label{sec:memory_exp}

Figure \ref{fig:memory_consumption} summarizes the memory consumption of the systems after inserting all edges, measured using the \emph{resident set size} (RSS) reported by the operating system. RapidStore is the most memory-efficient, saving up to 56.34\% of memory compared to other systems. This efficiency is achieved through ID compression and the avoidance of per-edge versioning. Such memory savings are particularly valuable for applications with limited resources, highlighting RapidStore's suitability for memory-constrained environments.

%% file: 9_related_work.tex
\section{Related Work}\label{sec:related_work}

\noindent\textbf{Tree-like Structures.} Tree-like structures are crucial for designing dynamic graph systems. Terrace~\cite{pandey2021terrace} employs a \emph{Structure Hierarchy} and uses B+ trees to store large neighborhoods efficiently. Additionally, binary search trees have been utilized for graph storage~\cite{dhulipala2022pac, wheatman2018packed}. Various radix tree variants have also been explored in graph contexts~\cite{ribeiro2014g}. These structures enhance the performance of dynamic graph operations by optimizing data retrieval and storage.

\vspace{2pt}
\noindent\textbf{Graph Systems.} The pioneering works Ligra~\cite{shun2013ligra} and Ligra+\cite{shun2015smaller} set the stage for high static read performance in graph processing. Other systems, such as Stinger~\cite{ediger2012stinger} and NXGraph~\cite{chi2016nxgraph}, implement \emph{Segmentation} techniques to improve graph processing efficiency. However, these early systems face significant challenges in efficiently handling concurrent read and write operations. In contrast, LLAMA~\cite{mccoll2014performance} introduces delta Storage to minimize update costs while ensuring snapshot isolation. Pensieve~\cite{ying2020pensieve} further enhances performance by applying delta storage to high-degree vertices, effectively balancing read and write operations. Similarly, GraphOne~\cite{kumar2020graphone} adopts a logging approach for updates, transforming them into a compact structure. LSMGraph~\cite{yu2024lsmgraph} is a disk-based system that combines LSM Tree and CSR, providing efficient read and write. Its in-memory part uses CSR and skip lists to store low- and high-degree vertices. Dynamic graph systems such as KickStarter~\cite{vora2017kickstarter}, GraphBolt~\cite{mariappan2019graphbolt}, and RisGraph~\cite{feng2021risgraph} are specifically designed for continuous graph processing. These systems enhance continuous graph algorithms~\cite{li2020space, zhang2022online} by storing intermediate results during computation, significantly reducing re-computation costs after graph evolution.

\begin{figure}[t]
	\setlength{\abovecaptionskip}{0pt}
	\setlength{\belowcaptionskip}{0pt}
		\captionsetup[subfigure]{aboveskip=0pt,belowskip=0pt}
	\centering
    \includegraphics[width=0.45\textwidth]{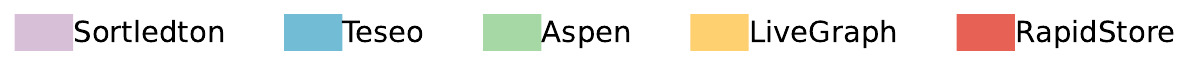}
    \\
	\begin{subfigure}[t]{0.45\textwidth}
		\centering
		\includegraphics[width=\textwidth]{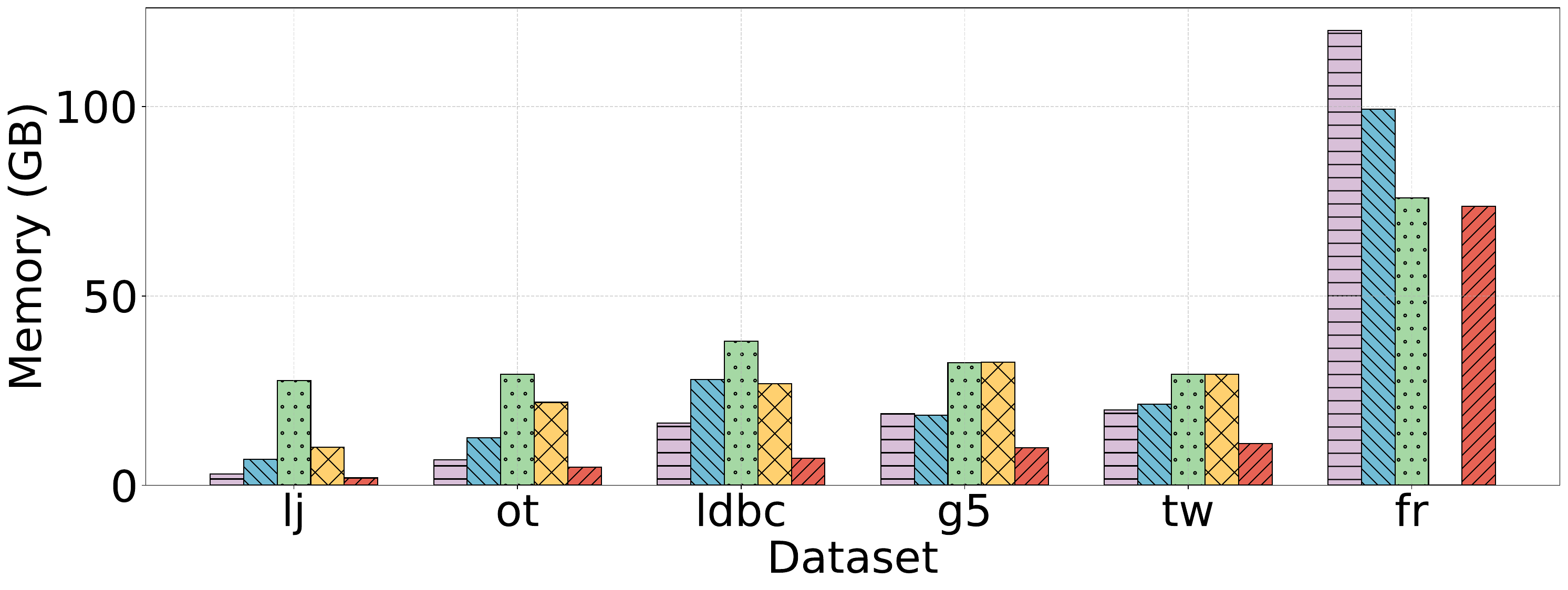}
	\end{subfigure}
	\caption{\sun{Memory consumption of systems.}}
	\label{fig:memory_consumption}
\end{figure}

\vspace{2pt}
\noindent\textbf{Graph Databases.} Several graph databases, including Neo4J, Virtuoso, GraphFlow~\cite{kankanamge2017graphflow}, and Kùzu~\cite{kuzu:cidr}, adopt specialized internal data structures (e.g., linked lists) to store and traverse graph data efficiently. An alternative approach augments relational databases with additional indexes~\cite{bronson2013tao, corbett2013spanner, shen2023bridging}, offering flexibility for diverse workloads but often at the cost of lower performance on graph-specific operations. Our work improves the in-memory dynamic graph storage through system-level optimizations such as subgraph-centric MVCC and C-ART, which reduce contention and improve memory access locality. These techniques could be integrated into existing graph systems to improve both throughput and concurrency. Furthermore, RapidStore’s decoupled architecture and memory-efficient indexing offer promising directions for building high-performance, general-purpose graph storage engines.

%% file: 10_conclusion.tex
\section{Conclusion} \label{sec:conclusion}
In this paper, we present RapidStore, an efficient in-memory dynamic graph storage system optimized for concurrent queries. RapidStore introduces a decoupled design that separates the management of read and write queries while isolating version data from graph data. This is achieved through a subgraph-centric concurrency control and an efficient multi-version graph store.  RapidStore significantly reduces query execution times compared to existing solutions, maintains high concurrency performance with very well hardware utilization, and offers competitive write performance with substantially lower memory consumption. These features establish RapidStore as a robust and effective dynamic graph data management solution. \sun{RapidStore is a high-performance in-memory dynamic graph storage system. Extending it to support persistence on emerging storage media, such as NVMe SSDs and cloud storage, is a promising direction for future research.}

%% file: 11_appendix.tex
\appendix
\clearpage

\section{Proof of Propositions} \label{appendix:proof-of-propositions}

\subsection{Correctness of Concurrency Protocol}
\begin{proposition}
The subgraph-centric concurrency control mechanism guarantees the serializability of both write and read queries.
\end{proposition}

\begin{proof}
To prove that our concurrency control mechanism ensures serializability, we need to demonstrate that the execution of concurrent read and write queries under this mechanism is equivalent to some serial execution of these queries. We will show this by examining the behavior of write queries, read queries, and their interactions.

\noindent\textbf{Serializability of Write Queries}

\begin{enumerate}
    \item Exclusive Locking with MV2PL: Write queries employ Multi-Version Two-Phase Locking (MV2PL) to synchronize updates. Specifically, a write query $ W_0 $ intending to update a set of vertices $ \Delta V $ identifies the set of subgraphs $ \Delta \mathcal{S} $ that contain these vertices. It then acquires exclusive locks on these subgraphs in ascending order of their subgraph IDs. This consistent locking order prevents deadlocks and ensures that concurrent write queries serialize their access to shared subgraphs.

    \item Commit Order Enforcement: After applying updates and creating new subgraph snapshots, $ W_0 $ atomically increments the global write timestamp $ t_w $ to obtain its commit timestamp $ t $. It then assigns $ t $ to the new versions of the modified subgraphs and links them to the heads of their respective version chains.

    \item Advancing the Read Timestamp: $ W_0 $ polls the global read timestamp $ t_r $. If $ t_r = t - 1 $, it atomically increments $ t_r $ by 1. This step ensures that write queries commit in the serial order determined by their commit timestamps. The condition $ t_r = t - 1 $ enforces that writes with earlier timestamps have already advanced $ t_r $, thereby preventing out-of-order commits.

    \item Serial Equivalence: Since write queries acquire exclusive locks and commit in a total order defined by their commit timestamps, the effect of executing concurrent write queries under this protocol is equivalent to executing them serially in the order of their commit timestamps. There are no write-write conflicts because locks prevent simultaneous updates to the same subgraphs.
\end{enumerate}

\noindent\textbf{Serializability of Read Queries}

\begin{enumerate}
    \item Snapshot Isolation: Read queries do not acquire locks and are not blocked by write queries. When a read query $ R $ begins, it registers itself in the reader tracer and records its start timestamp $ t $ as the current global read timestamp $ t_r $.

    \item Consistent Snapshot Construction: $ R $ constructs its graph snapshot view by selecting the latest subgraph snapshots whose versions are less than or equal to its start timestamp $ t $. Since write queries only advance $ t_r $ after committing all their updates, $ R $ is guaranteed to see a consistent state of the graph as of time $ t $.

    \item Non-Interference with Writes: Write queries use a copy-on-write strategy to create new subgraph snapshots. This means that existing snapshots remain unchanged and accessible to read queries. Therefore, reads and writes do not interfere with each other, and reads do not observe partial effects of concurrent writes.
\end{enumerate}

\noindent\textbf{Combined Serializability}

\begin{enumerate}
    \item Equivalent Serial Execution Order: The serialization of write queries via their commit timestamps and the snapshot isolation provided to read queries together ensure that the execution is equivalent to some serial order where:
    \begin{itemize}
        \item Write Queries: Executed sequentially in the order of their commit timestamps.
        \item Read Queries: Each read query observes the state of the graph after all write queries with commit timestamps less than or equal to its start timestamp $ t $ have been executed.
    \end{itemize}

    \item Absence of Anomalies: The mechanism prevents common anomalies such as dirty reads, non-repeatable reads, and lost updates:
    \begin{itemize}
        \item Dirty Reads: Read queries do not see uncommitted data because they only access subgraph snapshots with versions less than or equal to $ t $, and $ t $ is based on the committed $ t_r $.
        \item Non-Repeatable Reads: Since read queries operate on immutable snapshots, repeated reads within the same query return consistent data.
        \item Lost Updates: Write queries acquire exclusive locks and commit in a serial order, preventing overwriting of concurrent updates.
    \end{itemize}

    \item Atomic Operations and Lock-Free Reads: The use of atomic operations for timestamp increments and reader tracer updates ensures thread safety without introducing significant overhead. The ability of read queries to proceed without locks enhances concurrency while maintaining consistency.
\end{enumerate}

\noindent\textbf{Conclusion}

By serializing write queries through MV2PL and ensuring that read queries operate on consistent snapshots corresponding to specific points in the serial execution order, the proposed concurrency control mechanism guarantees that the concurrent execution of read and write queries is equivalent to some serial execution. Therefore, the mechanism ensures the serializability of both write and read queries.
\end{proof}

\setlength{\textfloatsep}{0pt}
\begin{figure*}[t]
	\setlength{\abovecaptionskip}{0pt}
	\setlength{\belowcaptionskip}{0pt}
		\captionsetup[subfigure]{aboveskip=0pt,belowskip=0pt}
	\centering
    \includegraphics[width=0.45\textwidth]{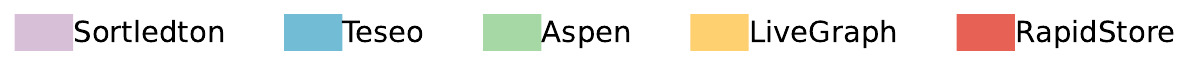}
    \\
	\begin{subfigure}[h]{0.23\textwidth}
		\centering
		\includegraphics[width=\textwidth]{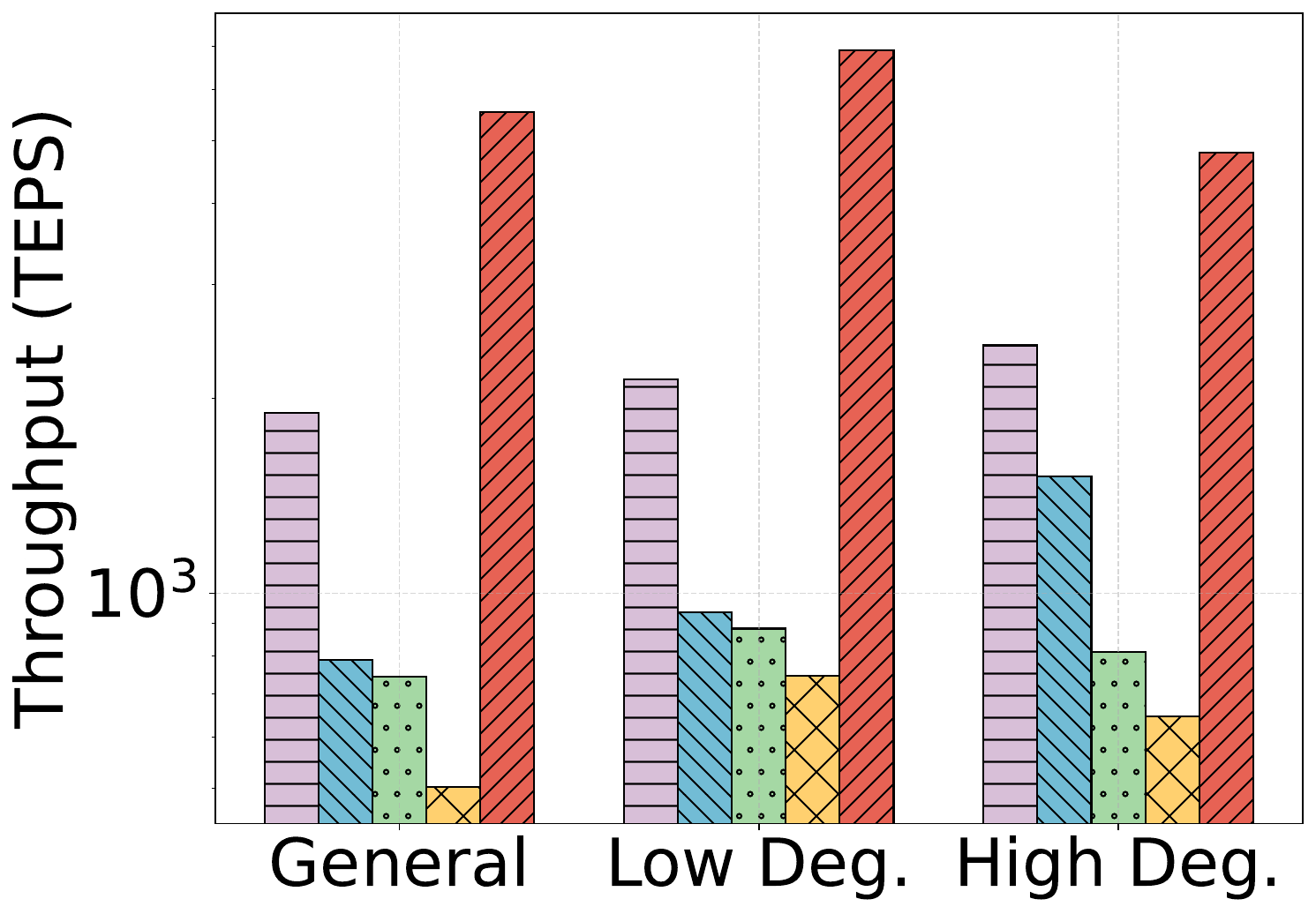}
		\caption{\emph{Search} on \textit{lj}.}
		\label{fig:search_lj}
	\end{subfigure}
        \begin{subfigure}[h]{0.23\textwidth}
		\centering
		\includegraphics[width=\textwidth]{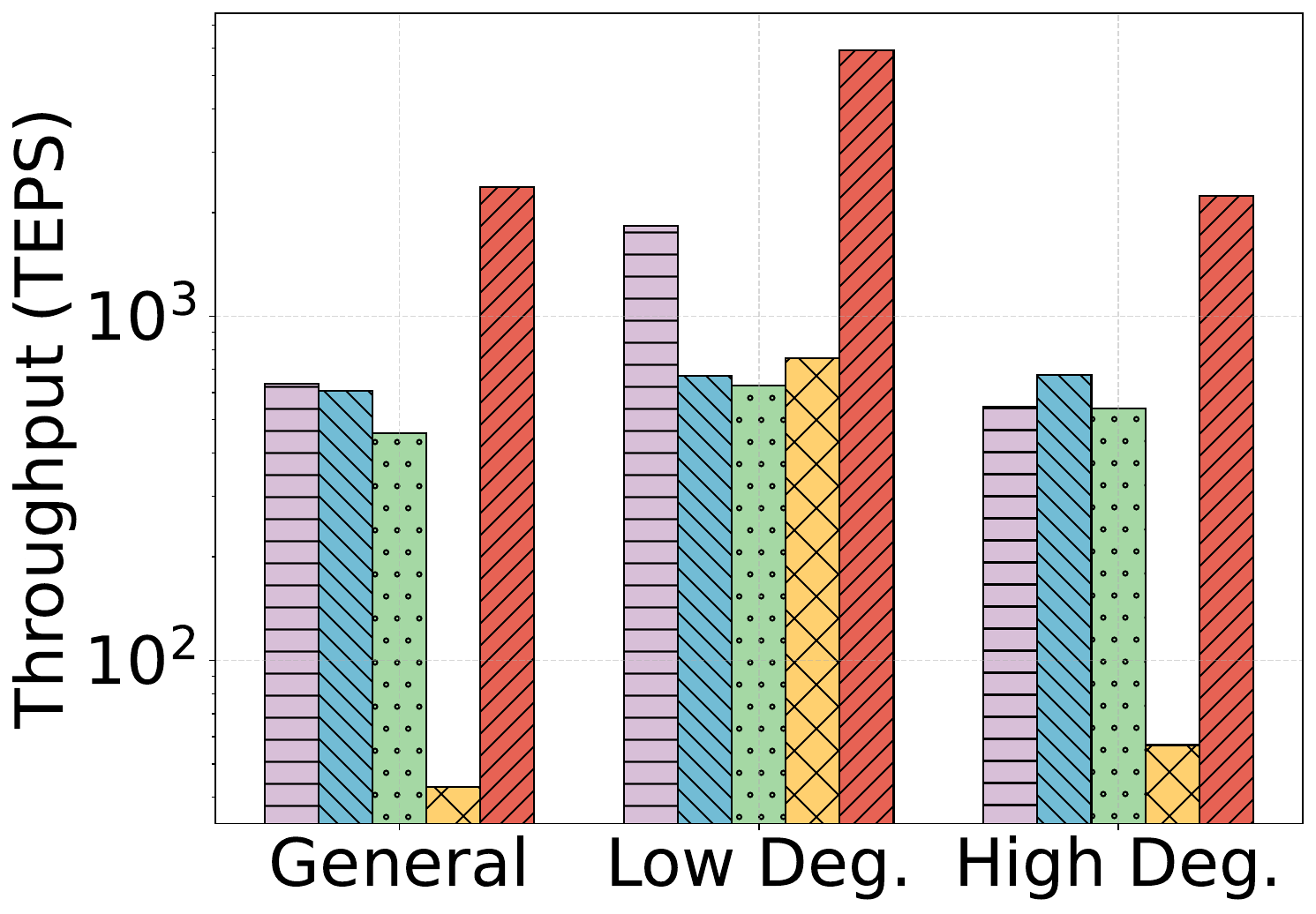}
		\caption{\emph{Search} on \textit{g5}.}
		\label{fig:search_g5}
	\end{subfigure}
        \begin{subfigure}[h]{0.23\textwidth}
		\centering
		\includegraphics[width=\textwidth]{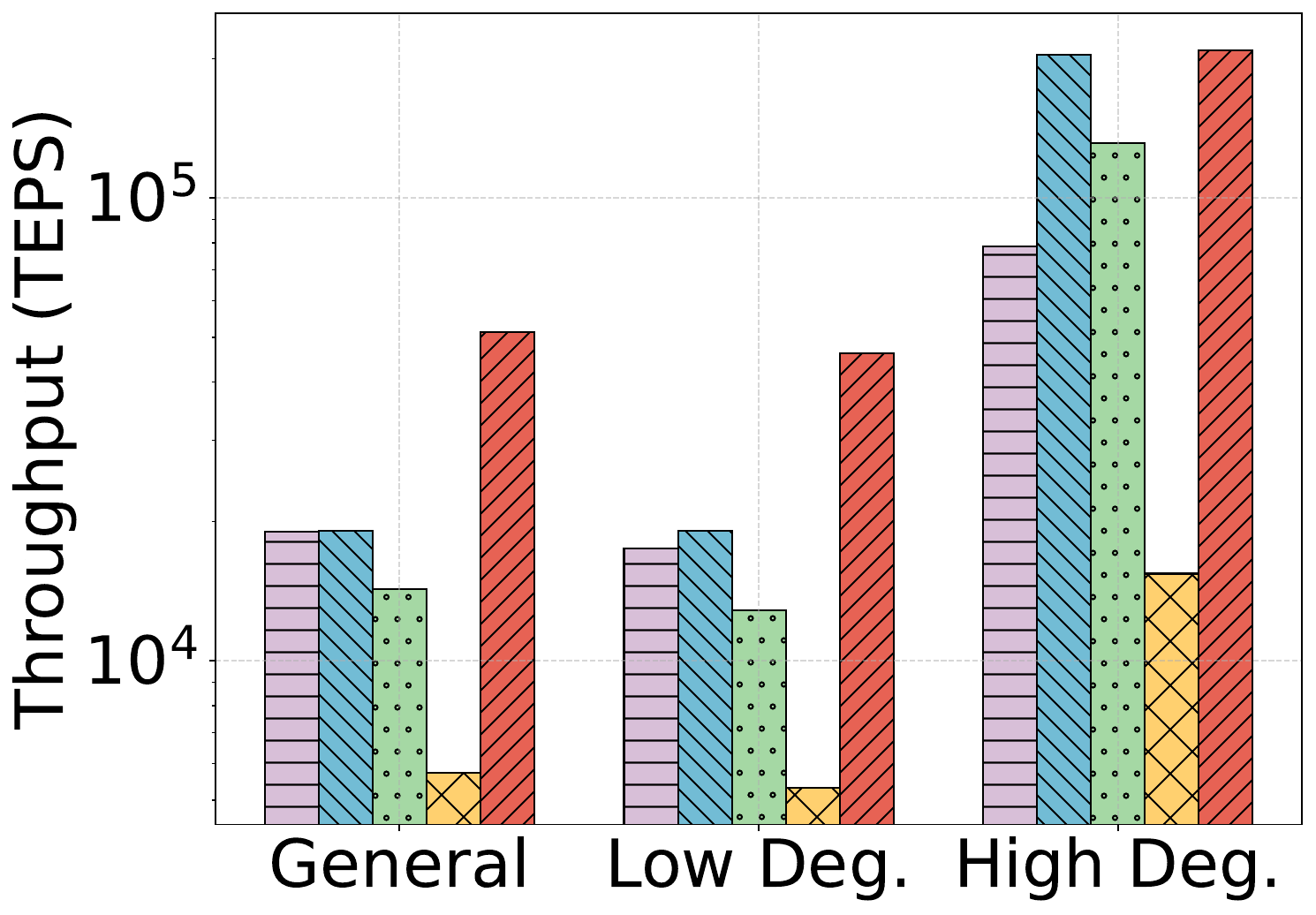}
		\caption{\emph{Scan} on \textit{lj}.}
		\label{fig:scan_lj}
	\end{subfigure}
        \begin{subfigure}[h]{0.23\textwidth}
		\centering
		\includegraphics[width=\textwidth]{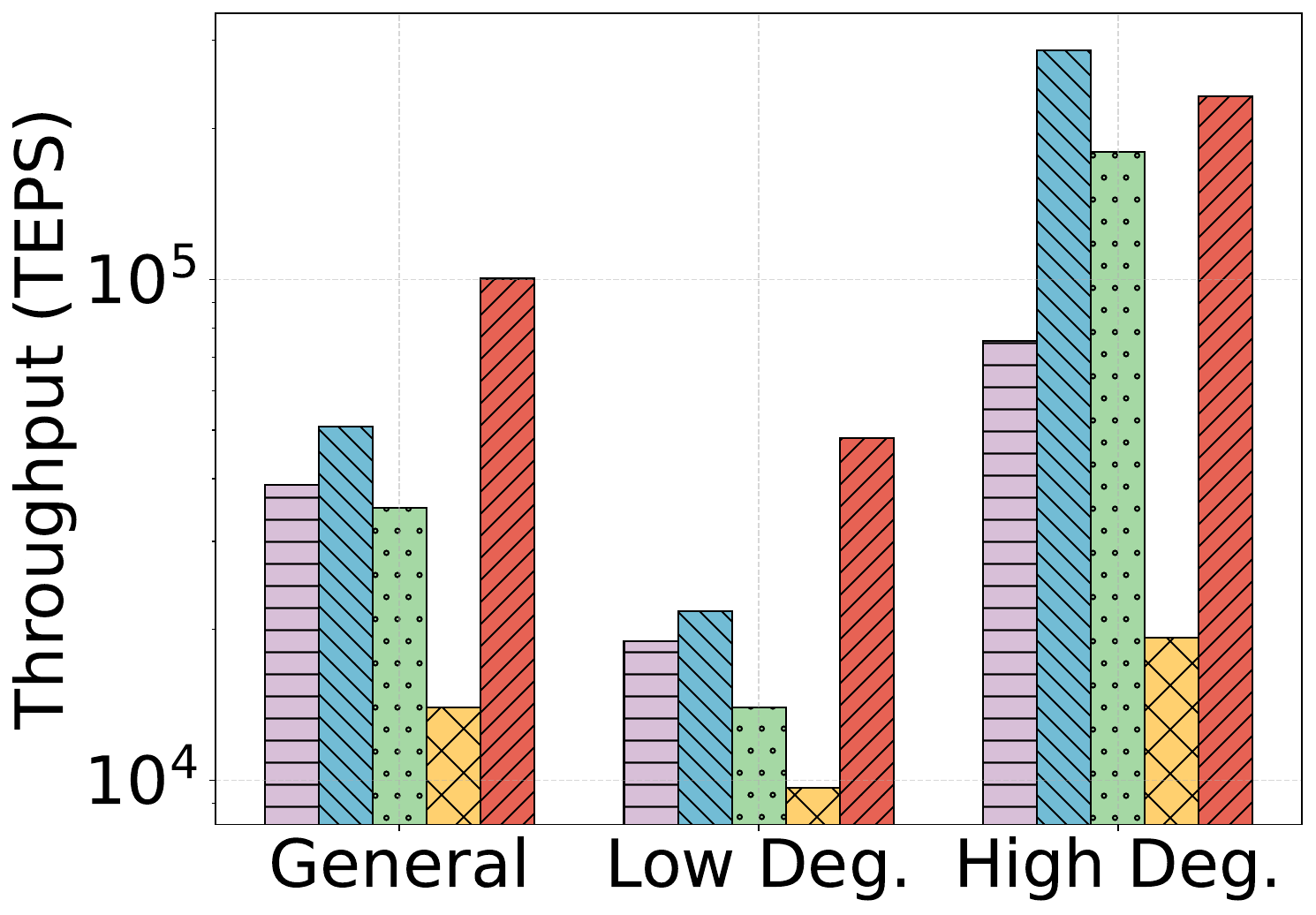}
		\caption{\emph{Scan} on \textit{g5}.}
		\label{fig:scan_g5}
	\end{subfigure}
	\caption{Performance of basic read operations.}
	\label{fig:scan}
\end{figure*}

\subsection{Maximum Length of Version Chain}
\begin{proposition} \label{prop:length_version_appendix}
	Given a subgraph $S$, the length of its version chain is at most $k + 1$, where $k$ is the size of the reader tracer array, representing the maximum number of concurrent read queries.
	\end{proposition}
	
\begin{proof}
	At any given time, the version chain of subgraph $S$ consists of versions that are either currently in use by active read queries or have not yet been reclaimed by the garbage collection (GC) process. We will show that the total number of such versions is at most $k + 1$.
	
	\noindent\textbf{Exclusive Extension by Writers} 
	
	The concurrency control mechanism ensures that only one writer can modify $S$ at a time because write queries acquire exclusive locks on subgraphs. Thus, the version chain of $S$ is extended by only one writer at any moment.

	\noindent\textbf{Active Readers Limitation}
	
	The reader tracer has a fixed size $k$, which is the maximum number of concurrent read queries the system supports. Therefore, there can be at most $k$ active read queries holding references to versions of $S$.

	\noindent\textbf{Garbage Collection Process}
	
	After a writer $W$ extends the version chain by adding a new version of $S$, it performs garbage collection. $W$ scans the reader tracer to identify the start timestamps of all active read queries. Using this information, $W$ determines which versions of $S$ are no longer needed by any reader.

	\noindent\textbf{Versions in Use}
	
	Since there are at most $k$ active readers, there are at most $k$ versions of $S$ that are currently in use and cannot be reclaimed during GC.

	\noindent\textbf{Total Versions in the Chain}
	
	In addition to the $k$ versions potentially held by active readers, there is the new version just added by the writer $W$. This version may not yet be in use by any reader but exists at the head of the version chain.

	\noindent\textbf{Conclusion}
	
	Combining the maximum of $k$ versions held by active readers and the one new version added by the writer, the total length of $S$’s version chain is at most $k + 1$.

\end{proof}

\section{Supplement Experiment Results}

\subsection{Evaluation on Basic Read Operation.} \label{sec:basic_read_op_eval}

This experiment evaluates two basic graph operations: \textit{Search} and \textit{Scan}, under the following workload scenarios:

\begin{itemize}[leftmargin=*]
    \item \textbf{General:} Vertices are selected with equal probability $\frac{1}{|V|}$.
    \item \textbf{Low-Degree:} The top 10\% of low-degree vertices are selected.
    \item \textbf{High-Degree:} The top 10\% of high-degree vertices are selected.
\end{itemize}

\vspace{2pt}
\noindent\textbf{Evaluation of \emph{Search}.} Figures \ref{fig:search_lj} and \ref{fig:search_g5} show the \emph{search} performance. Throughput is generally higher on \emph{lj}, a sparser graph. RapidStore consistently delivers the best performance across all settings, leveraging C-ART's constant \emph{search} complexity. LiveGraph performs poorly due to its unsorted neighborhoods, with slightly better results in the \textit{Low-Degree} setting. These results confirm that RapidStore's design effectively enhances \emph{search} throughput.

\vspace{2pt}
\noindent\textbf{Evaluation of \emph{Scan}.} Figures \ref{fig:scan_lj} and \ref{fig:scan_g5} show the \emph{scan} throughput. In the \textit{High-Degree} setting, RapidStore is slightly slower than Teseo, as Teseo stores neighbor segments contiguously, while RapidStore stores them separately. However, RapidStore outperforms competitors in the \textit{Low-Degree} setting with 1.22x--7.66x higher throughput and maintains a lead of 1.97x--7.97x in the \textit{General} setting. These results highlight the clustered index's effectiveness in enabling superior \emph{scan} performance for graph analytics. Sortledton is slower than Aspen due to version checks, while LiveGraph is the slowest because of its MVCC mechanism.

In summary, RapidStore provides efficient \emph{Search} and \emph{Scan} operations, making it highly suitable for graph analytics.

\subsection{Evaluation on Multicore Scalability}

Figure~\ref{fig:scalability} shows the scalability of each system from 1 to 32 writer threads. Overall, scalability is higher on \emph{lj} due to its more uniform vertex degree distribution, which reduces contention. Sortledton achieves the best scalability, reaching a 15.12x speedup with 32 writers, thanks to its edge-level MVCC and fully decoupled vertex neighborhoods. RapidStore ranks second with a 9.92x speedup, limited by vertex-group MVCC, which increases the scope of lock contention. Teseo achieves 9.25x but suffers from periodic PMA rebalancing that blocks all writers. LiveGraph performs the worst—its throughput drops beyond 16 writers due to costly global synchronization required by its MVCC.

\begin{figure}[t]
	\setlength{\abovecaptionskip}{0pt}
	\setlength{\belowcaptionskip}{0pt}
		\captionsetup[subfigure]{aboveskip=0pt,belowskip=0pt}
	\centering
    \includegraphics[width=0.47\textwidth]{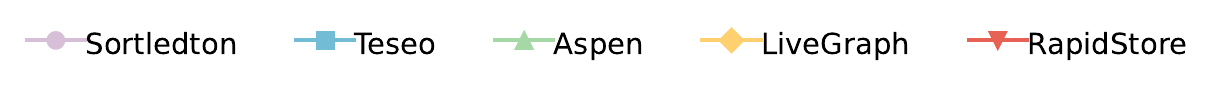}
	\begin{subfigure}[h]{0.23\textwidth}
		\centering
		\includegraphics[width=\textwidth]{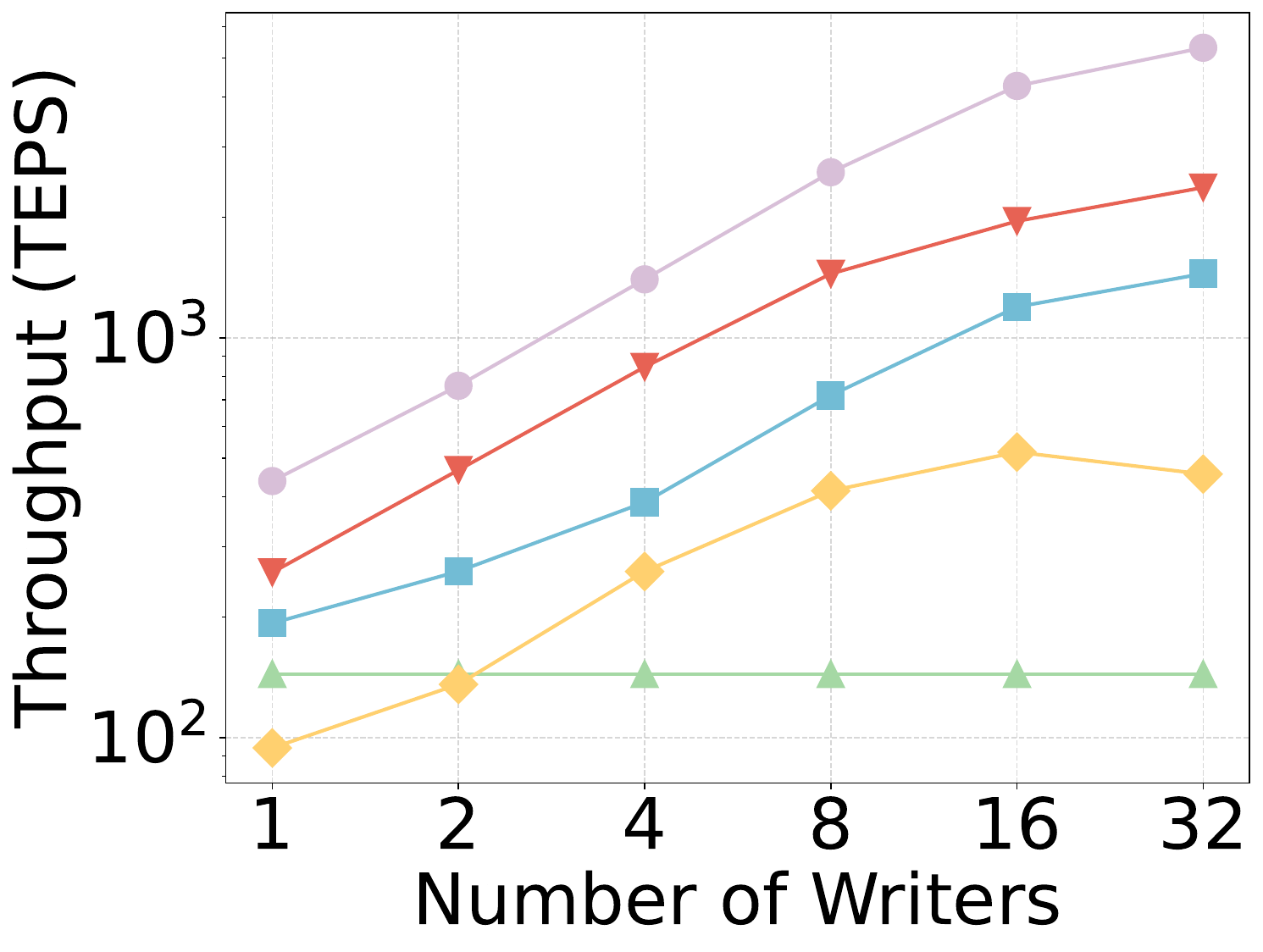}
		\caption{Insertion on \textit{lj}.}
		\label{fig:scalability_lj}
	\end{subfigure}
        \begin{subfigure}[h]{0.23\textwidth}
		\centering
		\includegraphics[width=\textwidth]{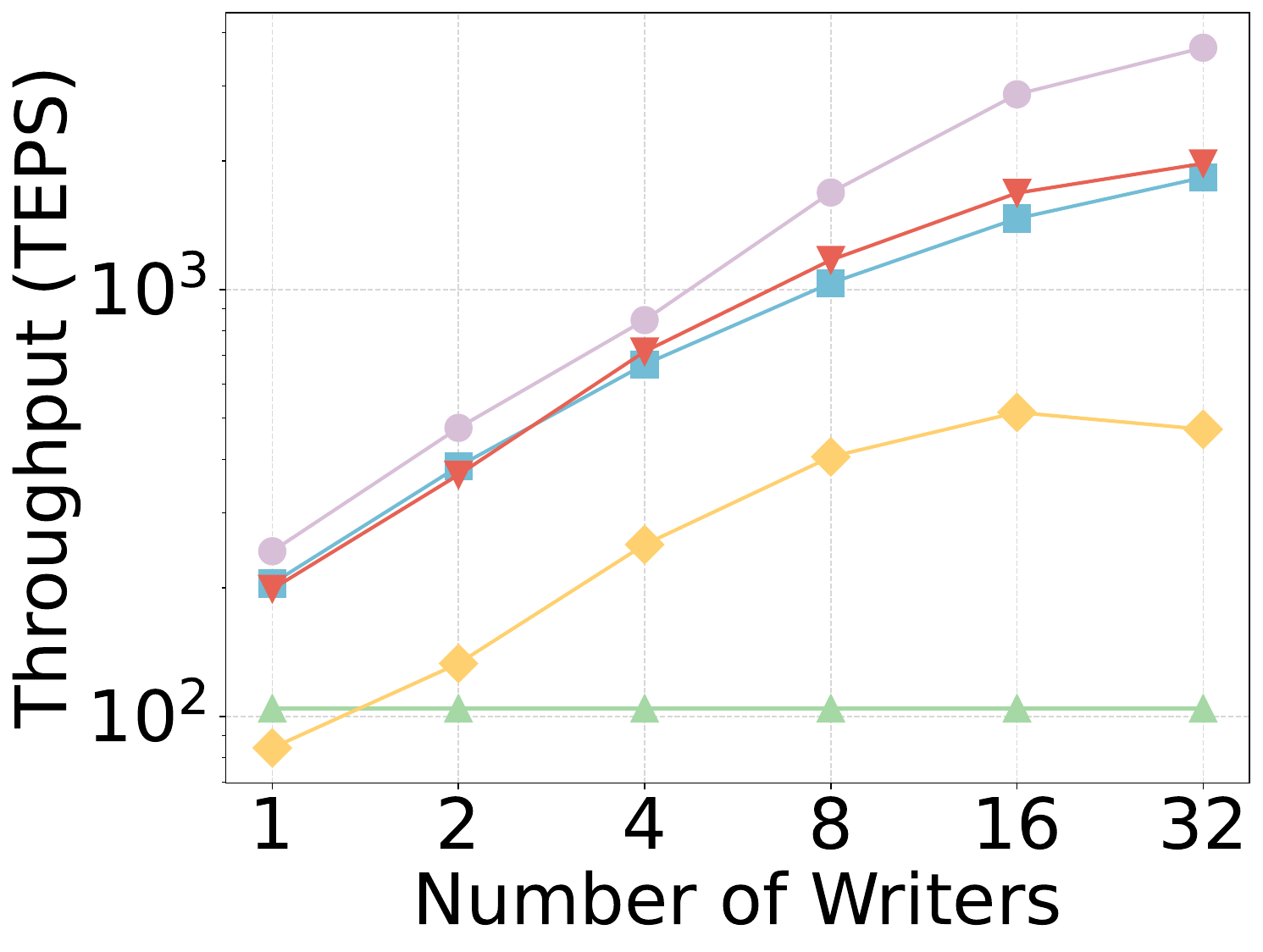}
		\caption{Insertion on \textit{g5}.}
		\label{fig:scalability_g5}
	\end{subfigure}
	\caption{Scalability with the number of writers varying.}
	\label{fig:scalability}
\end{figure}    

\subsection{Evaluation on Batch Update}

\begin{figure}[h]
	\setlength{\abovecaptionskip}{0pt}
	\setlength{\belowcaptionskip}{0pt}
		\captionsetup[subfigure]{aboveskip=0pt,belowskip=0pt}
	\centering
	\begin{subfigure}[t]{0.22\textwidth}
		\centering
		\includegraphics[width=\textwidth]{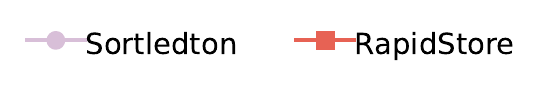}
	\end{subfigure}
	\begin{subfigure}[t]{0.45\textwidth}
		\centering
		\includegraphics[width=\textwidth]{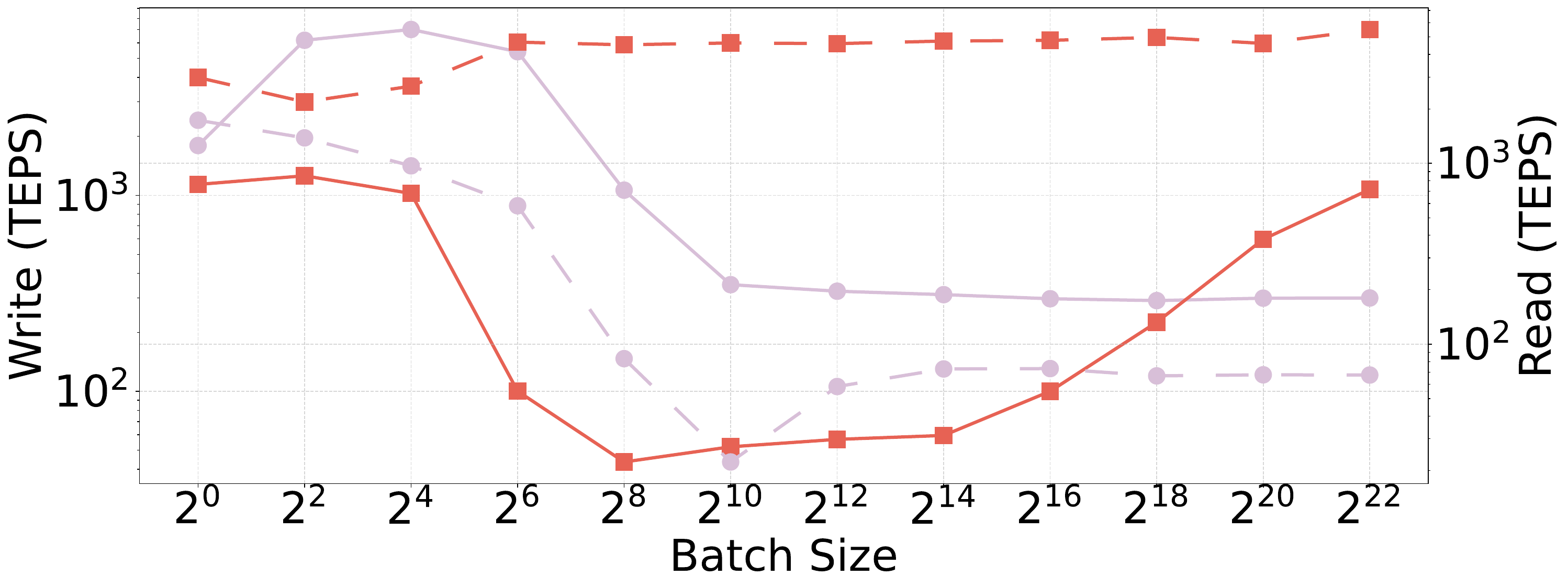}
        \caption{Performance on \emph{lj}}
	\end{subfigure}
	\begin{subfigure}[t]{0.45\textwidth}
		\centering
		\includegraphics[width=\textwidth]{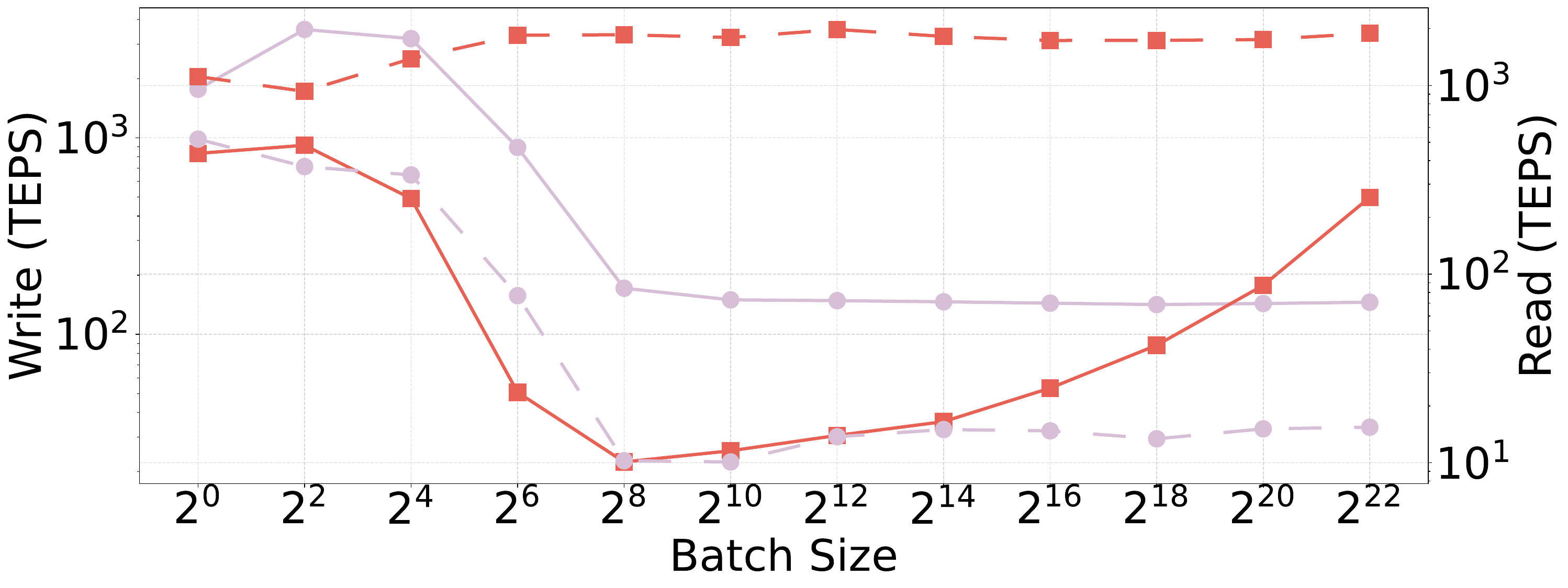}
        \caption{Performance on \emph{g5}}
	\end{subfigure}
	\caption{\new{Evaluation of large writes with small lookups under varying batch sizes. The solid line represents write throughput (Batch Update), and the dashed line represents read throughput (Search), both measured in thousand edges per second (TEPS).}}
	\label{fig:batch_update}
\end{figure}
\new{RapidStore focuses on common read-intensive workloads in the real world--- large analytics queries + small writes. To evaluate it in the opposite scenario, we conduct an experiment where 32 threads are launched: one thread continuously issues search operations, and the remaining 31 threads execute batch updates (i.e., each update consists of multiple update operations). LiveGraph and Teseo are excluded as their implementations do not support batch update isolation and encounter execution errors. Aspen is excluded due to its limited write performance.}

\new{Figure~\ref{fig:batch_update} presents the throughput results as the batch size varies. When the batch size is small (batch size = 1), RapidStore exhibits slightly lower write throughput compared to Sortledton, primarily due to its coarser-grained versioning at the subgraph level, as detailed in the initial submission. As the batch size increases ($2^2$ to $2^{10}$), Sortledton’s write throughput degrades due to rising lock contention. RapidStore’s write performance initially drops for the same reason, but then improves and eventually surpasses Sortledton. This is because larger batches allow RapidStore’s copy-on-write mechanism to share more work within a subgraph, outweighing the lock contention overhead. RapidStore consistently outperforms Sortledton for read operations across all batch sizes with at most 141.2x higher throughput.}

\new{In summary, RapidStore delivers superior read performance and good write performance, particularly for large batch updates. These results further validate RapidStore’s effectiveness across both read- and write-intensive scenarios.}

\subsection{Real LDBC Workload}

\new{In this section, we extracted the actual update/insert query sequence from the LDBC\_SNB\_Interactive workload and executed it on all evaluated systems to better evaluate their performance under a real workload. Figure~\ref{fig:ldbc_real} displays performance results comparing random insert workload against real update workload. }

\new{For per-edge versioning systems, Sortledton's throughput increases by 20.89\% with the real workload, indicating its vulnerability to lock contention during random insertions. Teseo shows no significant performance variation between workloads. LiveGraph suffers a dramatic 78.92\% throughput reduction in the real workload scenario, attributable to many update/insert queries targeting high-degree vertices where its insert operations slow considerably due to expensive neighbor set searches (as previously analyzed in Section 7.2.1).}

\new{Copy-on-write approaches exhibit different patterns: both Aspen and RapidStore experience throughput decreases (15.87\% and 14.14\%, respectively) with the real workload. Since Aspen cannot be affected by lock contention, its performance decline comes from changes in the copy process, which is the same factor affecting RapidStore. Most systems maintain relatively consistent performance across both workloads, with LiveGraph being the notable exception. }

\new{While insertion order changes moderately impact throughput across systems, their relative performance ranking remains stable, with RapidStore consistently demonstrating strong performance in both workloads.
}

\begin{figure}[t]
	\setlength{\abovecaptionskip}{0pt}
	\setlength{\belowcaptionskip}{0pt}
		\captionsetup[subfigure]{aboveskip=0pt,belowskip=0pt}
	\centering
	\begin{subfigure}[t]{0.45\textwidth}
		\centering
		\includegraphics[width=\textwidth]{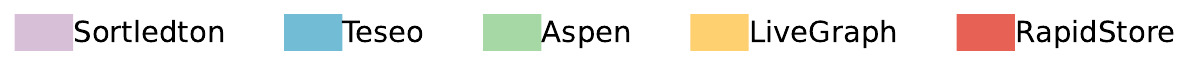}
	\end{subfigure}
	\begin{subfigure}[t]{0.45\textwidth}
		\centering
		\includegraphics[width=\textwidth]{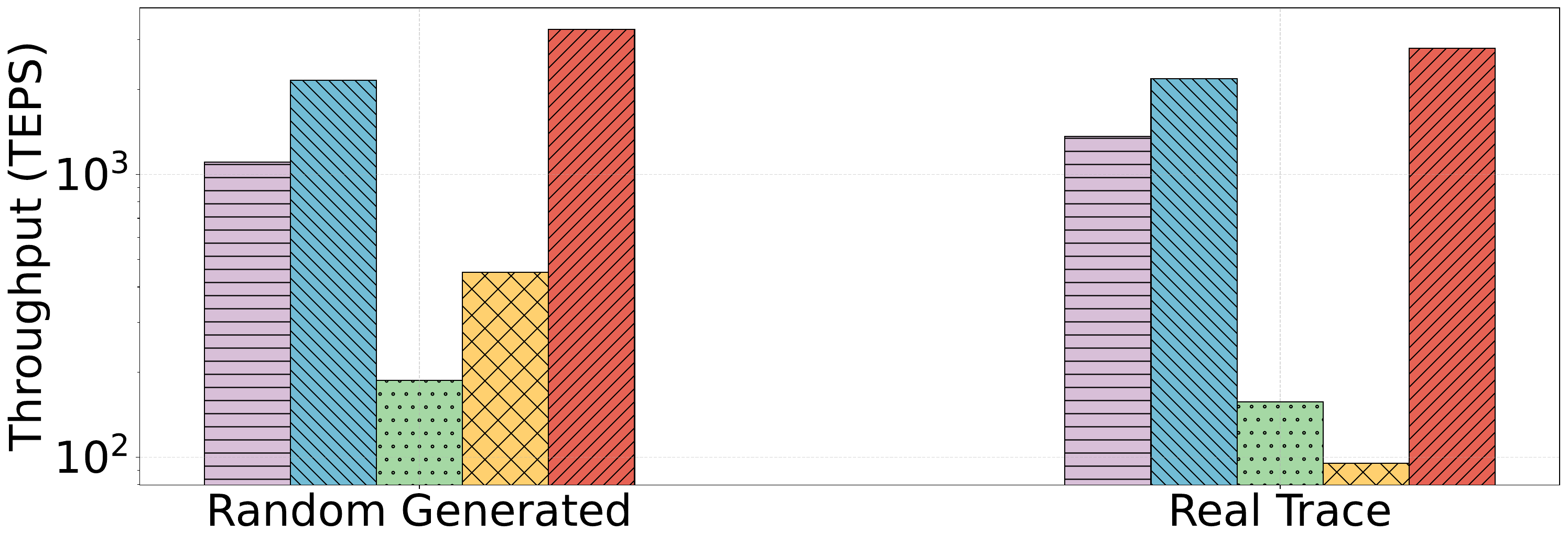}
	\end{subfigure}
	\caption{\new{Write throughput of the evaluated systems on the \emph{ldbc} dataset. The left plot reports results using a randomly generated trace, while the right plot uses the real \emph{ldbc} trace.}}
	\label{fig:ldbc_real}
\end{figure}
\subsection{Insertion over Growing Neighbor}

\begin{figure}[t]
	\setlength{\abovecaptionskip}{0pt}
	\setlength{\belowcaptionskip}{0pt}
		\captionsetup[subfigure]{aboveskip=0pt,belowskip=0pt}
	\centering
    \includegraphics[width=0.50\textwidth]{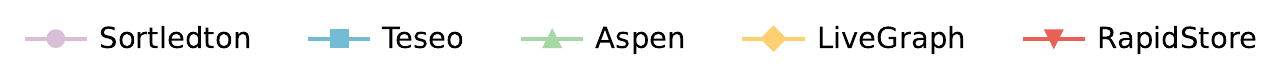}
    \begin{subfigure}[h]{0.45\textwidth}
		\centering
		\includegraphics[width=\textwidth]{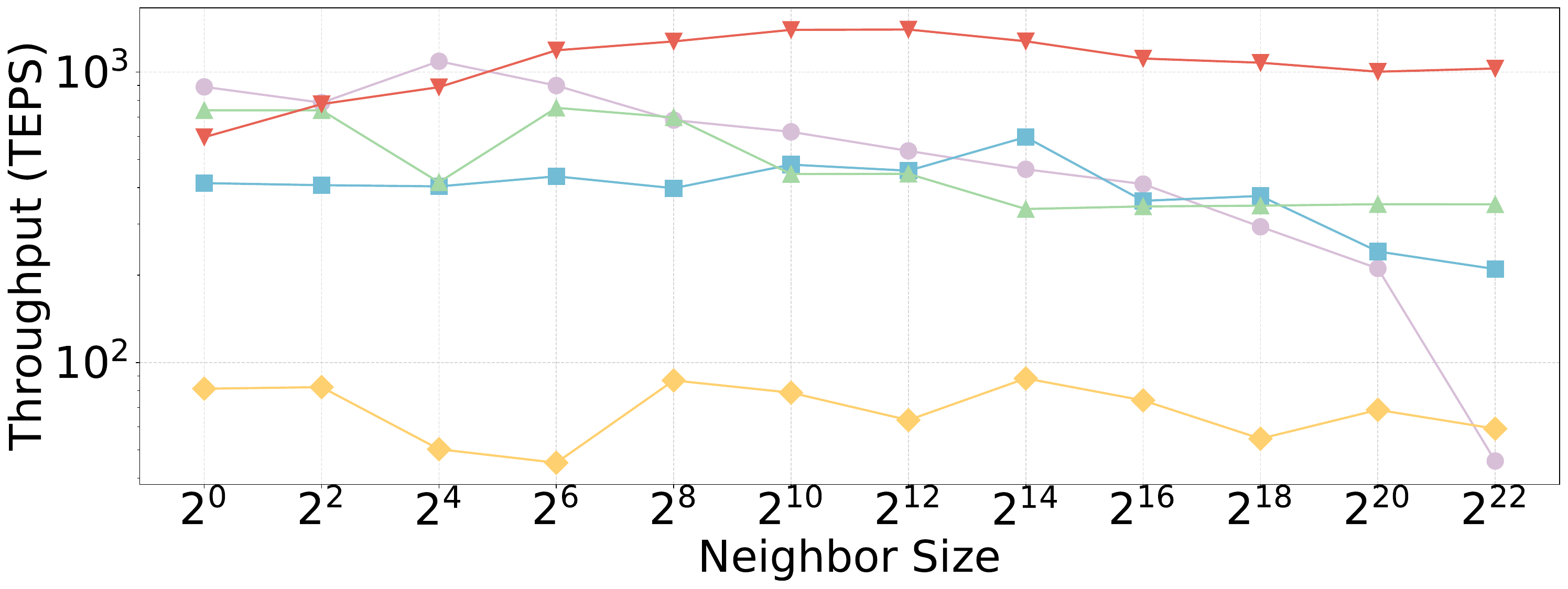}
	\end{subfigure}
   \caption{Performance of edge insertion with growing neighbor size.}
	\label{fig:growing_neighbor_write}
\end{figure}

This experiment evaluates the performance of different systems inserting neighbors of different sizes $|N|$, aiming to evaluate their versioning methods and data structures. For different $|N|$, we divide the $2^{24}$ edges into vertices of the same size $|N|$ and insert them after shuffling with 1 writer. 

Figure ~\ref{fig:growing_neighbor_write} shows the experiment results. Sortledton performs well for $|N|$ less than 16. However, Sortledton's search time increases as $|N|$ becomes larger, making the insertion speed decrease continuously, eventually by 94.85\% compared to when $|N|=2^0$. The same applies to Aspen, which has a 52.56\% slowdown. For Teseo, although its use of ART additional indexes avoids search slowdowns, the cost of its rebalance becomes higher as $|N|$ increases, so there is also a 49.42\% slowdown. Unlike other systems, RapidStore's ability to keep insertion speeds stable due to the constant-level search complexity provided by C-ART means that it can be applied to graphs of all shapes.

\clearpage